%% file: main.tex
\pdfoutput=1 
\documentclass[fleqn,11pt,letter]{article}
\usepackage[utf8]{inputenc}
\usepackage[pdftex]{graphicx}
\usepackage[usenames,dvipsnames,table]{xcolor}
\usepackage{fullpage}
\usepackage{times}

\usepackage{wrapfig}
\usepackage{sidecap}
\usepackage{enumitem}

\newcommand{\redcol}[1]{\textcolor{Mahogany}{#1}}
\newcommand{\greencol}[1]{\textcolor{OliveGreen}{#1}}
\newcommand{\redu}{\redcol{u}}
\newcommand{\greenv}{\greencol{v}}

\newcommand{\greenw}{\greencol{w}}

\usepackage{amssymb}
\usepackage{amsmath}
\usepackage{float}
\usepackage{morefloats}
\usepackage{graphicx}
\usepackage{amsthm}
\usepackage{rotating}
\usepackage{wrapfig}

\usepackage{cite} 

\usepackage{hyperref}

\usepackage{url} 
\urldef{\maildiku}\path|roden@diku.dk|
\urldef{\mailjacob}\path|jaho@di.ku.dk|

\usepackage{todonotes}
\DeclareGraphicsExtensions{.pdf,.png,.jpg,.jpeg,.mps}

\newcommand{\eps}{\varepsilon}

\newcommand{\Oo}{\mathcal O} 
\newcommand{\lca}{\operatorname{lca}}

\newcommand{\dual}[1]{\ensuremath{{#1}^*}}

\newcommand{\source}{\operatorname{init}}

\newcommand{\set}[1]{\left\{#1\right\}}
\newcommand{\sizeof}[1]{\left\lvert#1\right\rvert}
\newcommand{\simplepath}[2]{#1\rightsquigarrow{}#2}

\newcommand{\reachable}[2]{#1\rightsquigarrow{}#2}

\usepackage{amsthm}
\theoremstyle{plain}
\newtheorem{theorem}{Theorem}[section]

\theoremstyle{definition}
\newtheorem{lemma}[theorem]{Lemma}
\newtheorem{corollary}[theorem]{Corollary}

\newtheorem{definition}[theorem]{Definition}

\usepackage{authblk}

\author{Jacob Holm and Eva Rotenberg and Mikkel Thorup\thanks{Research partly supported by Thorup's
    Advanced Grant from the Danish Council for Independent Research
    under the Sapere Aude research career programme.} \thanks{Research partly supported by the FNU project AlgoDisc - Discrete Mathematics, Algorithms, and Data
    Structures.}}
\affil{University of Copenhagen (DIKU),\\
    jaho@di.ku.dk, roden@di.ku.dk, mthorup@di.ku.dk} 
\title{Planar Reachability in Linear Space and Constant Time}

\begin{document}
\thispagestyle{empty}
\maketitle
\begin{abstract}
We show how to represent a planar digraph in linear space so
that reachability queries can be answered in constant time. The data structure can be constructed in linear time. This representation
of reachability is thus optimal in both time and space, and has optimal construction time. 
The previous best solution used $O(n\log n)$ space for constant query
time [Thorup FOCS'01].
\end{abstract}

\thispagestyle{empty}
\newpage

\input{Intro.tex}
\input{Preliminaries.tex}
\input{SMPD.tex}

\input{4-frame.tex}

\input{InOut.tex}
\nocite{DBLP:journals/ipl/Kameda75}
\bibliographystyle{plain}
\bibliography{refs}
\newpage

\end{document}

%% file: Intro.tex
\section{Introduction}
\label{S-intro}
Representing reachability of a directed graph is a fundamental challenge.
We want to represent a digraph $G=(V,E)$, $n=|V|$, $m=|E|$, so that we
for any vertices $u$ and $w$ can tell if $u$ reaches $v$, that is, if
there is a dipath from $u$ to $v$. There are two extreme solutions: one
is to just store the graph, as is, using $O(m)$ words of space and answering
reachability queries from scratch, e.g., using breadth-first-search, in $O(m)$ time.
The other is to store a reachability matrix using $n^2$ bits and then
answer reachability queries in constant time. Thorup and Zwick \cite{TZ05:dist-oracle} proved
that there are graphs classes such that any representation of reachability needs
$\Omega(m)$ bits. Also, P\v{a}tra\c{s}cu~\cite{Pat11:unify-lower} has proved that there are directed
graphs with $O(n)$ edges where constant time reachability queries require $n^{1+\Omega(1)}$ space. Thus, for constant time reachability queries to a general digraph,
all we know is that the worst-case space is somewhere between $\Omega(m+n^{1+\Omega(1)})$ 
and $n^2$ bits.

The situation is in stark contrast to the situation for
undirected/symmetric graphs where we can trivially represent
reachability queries on $O(n)$ space and constant time, simply by
enumerating the connected components, and storing with each vertex the
number of the component it belongs to. Then $u$ reaches $v$ if and only
if the have the same component number.

In this paper we focus on the planar case, which feels particularly relevant when you live on a sphere. 
For planar digraphs it is already known that we can
do much better than for general digraphs. Back in 2001, Thorup~\cite{DBLP:journals/jacm/Thorup04}
presented a reachability oracle for planar digraphs 
using $O(n\lg n)$ space for constant query time, or linear space for $O(\log n)$ query time. 
In this paper, we present the
first improvement; namely an $O(n)$ space reachability oracle that can answer 
reachability queries in constant time. 
Note that this bound is asymptotically optimal; even to distinguish between the subclass 
of directed paths of length $n$, we need $\Omega(n\log n)$ bits.
Our oracle is constructed in linear time.

\paragraph{Computational model} 
The computational model for all upper bounds is the word RAM, modelling what we can program in a
standard programming language such as C~\cite{KR88}. A word
is a unit of space big enough to fit any vertex identifier, so
a word has $w\geq \lg n$ bits, and word operations take constant time. Here $\lg=\log_2$. In our upper bounds, we limit ourselves to the \emph{practical RAM} model~\cite{Mil96}, which is a restriction of the word RAM to the standard operations on words available in C that are AC0. This includes indexing arrays as needed
just to store a reachability matrix with constant time access, but excludes e.g. multiplication and division.
Thus, unless otherwise specified, we measure {\em
space\/} as the number of words used and {\em time\/} as the
number of word operations performed. 

The $\Omega(m+n^{1+\Omega(1)})$ space lower bound from 
\cite{Pat11:unify-lower} for general graphs is in the cell-probe model 
subsuming the word RAM with an arbitrary instruction set.

\paragraph{Other related work}
Before \cite{DBLP:journals/jacm/Thorup04}, the best reachability oracles for general planar digraphs were distance
oracles, telling not just if $u$ reaches $w$, but if so, also the length of
the shortest dipath from $u$ to $w$ \cite{CX00,ArChChDaSmZa96,Dji96}.
For such planar distance oracles, the best current time-space trade-off is $\tilde
O(n/\sqrt s)$ time for any $s\in[n,n^2]$ \cite{MS12}.

The construction of \cite{DBLP:journals/jacm/Thorup04} also yields approximate
distance oracles for planar digraphs. With edge weights from $[N]$, $N\leq 2^w$, 
distance queries where answered within a factor $(1+\epsilon)$ in
$O(\log\log(Nn)+1/\eps)$ time using $O(n(\log n)(\log(Nn))/\eps)$ space. These bounds have not been improved.

For the simpler case of undirected graphs, where reachability is
trivial, \cite{DBLP:journals/jacm/Thorup04,Kle02} provides a more efficient
$(1+\eps)$-approximate distance queries for planar graphs in
$O(1/\eps)$ time and $O(n(\log n)/\eps)$ space. In \cite{DBLP:journals/corr/abs-1104-5214} it was
shown that the space can be improved to linear if the query time is
increased to $O((\log n)^2/\eps^2)$. In \cite{KST13:undir-plan-oracle}
it was shown how to represent planar graphs with bounded weights using
$O(n \log^2((\log n)/\eps) \log^*(n) \log\log(1/\eps))$ space and
answering $(1+\eps)$ approximate distance queries in
$O((1/\eps)\log(1/\eps) \log \log(1/\eps)\log^*(n) + \log\log\log n))$
time. Using $\bar O$ to suppress factors of $O(\log\log n)$ and
$O(\log(1/\eps))$, these bounds reduce to $\bar O(n)$ space and $\bar
O(1/\eps)$ time. This improvement is similar in spirit to our
improvement for reachability in planar digraphs. However, the
techniques are entirely different.

There has also been work on special classes of planar
digraphs.  In particular, for a planar $s$-$t$-graph, where all
vertices are on dipaths between $s$ and $t$, Tamassia and Tollis
\cite{DBLP:journals/tcs/TamassiaT93} have shown that we can represent reachability in linear
space, answering reachability queries in constant time.  Also, 
\cite{CX00,DPZ91,DPZ95} presents improved bounds for
planar exact distance oracles when all the vertices are on the
boundary of a small set of faces.

\paragraph{Techniques}
We will develop our linear space constant query time reachability oracles
by considering more and more complex classes of planar digraphs.
We make reductions from $i+1$ to $i$ in the following:
\begin{enumerate}
\item Acyclic planar s-t-graph; $\exists (s,t)$, such that all vertices are reachable from $s$ and may reach $t$.~\cite{DBLP:journals/tcs/TamassiaT93}
\item Acyclic planar single-source graph; $\exists s$, such that all vertices are reachable from $s$. See Section~\ref{sec:smpd}.
\item Acyclic planar In-Out graph; $\exists s$ such that all vertices with out-degree $0$ are reachable from $s$.
See Section~\ref{sec:inout}
\item Any acyclic planar graph. The reduction to acyclic planar In-Out graphs from general acyclic planar graphs is known.~\cite{DBLP:journals/jacm/Thorup04}
\item Any planar graph. The reduction to acyclic planar graphs is
  well-known.  Using the depth first search algorithm by
  Tarjan~\cite{Tarjan72depthfirst}, we can contract each strongly
  connected component to get an acyclic planar graph.  Vertices in the
  same strongly connected component can always reach each other, and
  vertices in distinct strongly connected components can reach each
  other if the corresponding vertices in the contracted graph can.
\end{enumerate}



The most technically involved step is the reduction from single-source graph to s-t-graph. As in \cite{DBLP:journals/jacm/Thorup04}, we use separators to form a tree over a partitioning of the vertices of the graph. 
However, in~\cite{DBLP:journals/jacm/Thorup04}, the \emph{alternation number}; the number of directed segments in the frame that separates a child from its parent (see Section~\ref{sec:prelim}), needs only be a constant number. In contrast, it is a crucial part of our construction that the alternation number, which must be even, is at most $4$. Also, in our data structure, paths cannot go upward in the rooted tree, whereas there is no such restriction in~\cite{DBLP:journals/jacm/Thorup04}. These two features let us use a level ancestor -like algorithm to quickly calculate the best $\le 4$ vertices in a given tree-node that can reach a given vertex $v$. Each component is an s-t-graph, and $v$ can be reached by some $u$ in the ancestral component if and only if $u$ can reach at least one of these best $\le 4$ vertices.

%% file: Preliminaries.tex
\section{Preliminaries}\label{sec:prelim}

  For a vertex $v$ at depth $d$ in a rooted forest $T$ and an integer
  $0\leq{}i\leq{}d$, the $i$'th \emph{level ancestor} of $v$ in $T$ 
  is the ancestor to $v$ in $T$ at depth $i$.
For two nodes $x$, $y$ in a rooted tree, let $x \preceq y$ denote that $x$ is an ancestor to $y$, and $x \prec y$ that $x$ is a proper ancestor to $y$. 

We say a graph is \emph{plane}, if it is embedded in the plane, and denote by $\pi_v$ the permutation of edges around $v$.
%
Given a plane graph, $(G,\pi)$, we may introduce \emph{corners} to describe the incidence of a vertex to a face. A vertex of degree $n$ has $n$ corners, where if $\pi_v((v,u)) = (v,w)$, and the face $f$ is incident to $(v,u)$ and $(v,w)$, then there is a corner of $f$ incident to $v$ between $(v,u)$ and $(v,w)$.
We denote by $V[X]$ and $E[X]$ the vertices and edges, of some (not necessarily induced) subgraph $X$.
%
Given a subgraph $H$ of a planar embedded graph $G$, the faces of $H$ define \emph{superfaces} of those of $G$, and the faces of $G$ are \emph{subfaces} of those of $H$. Similarly for corners.
%
Note that the faces of $H$ correspond to the connected components of $G^\ast \setminus H$. The super-corners incident to $v$ correspond to a set of consecutive corners in the ordering around $v$.

In a directed graph, we may consider the boundary of a face in some  subgraph, $H$. A corner of a face $f$ of $H$ is a \emph{target} for $f$ if it lies between ingoing edges $(u,v)$ and $(w,v)$, and \emph{source} if it lies between outgoing edges $(v,u)$ and $(v,w)$.
We say the face boundary has \emph{alternation number} $2a$ if it has $a$ source and $a$ target corners.
When a face boundary has alternation number $2a$, we say it consists of $2a$ \emph{disegments} (\emph{di}rected segments), associated with the directed paths from source to target. We associate with each disegment the total ordering stemming from reachability of vertices on the path via the path, and by convention we set $\operatorname{succ}(t,S) = \bot$ for a target vertex $t$ on the disegment. 
Given a set of edges $S\subset E$, we denote by $\source(S)$ the set of inital vertices, $\source(S) = \{u|(u,v)\in S\}$.
%
%
Given a connected planar graph with a spanning tree $T$, the edges $T^\ast := E\setminus T$ form a spanning tree for the dual graph. We call the pair $(T,T^\ast)$ a tree-cotree decomposition of the graph, referring to $T$ and $T^\ast$ as tree and cotree.

When $u$ can reach $v$ we write $\simplepath{u}{v}$.
An s-t-graph is a graph with special vertices $s,t$ such that $\simplepath{s}{v}$ and $\simplepath{v}{t}$ for all vertices $v$. We say a graph is a \emph{truncated s-t-graph} if it is possible to add vertices $s,t$ to obtain an s-t-graph, without violating the embedding. In an s-t-graph, all faces has alternation number $2$.

%% file: SMPD.tex
\section{Acyclic planar single-source digraph}\label{sec:smpd}
Given a global source vertex $s$ for the planar digraph, we wish to
make a data structure for reachability queries. We do this by
reduction to the s-t-case.  A tree-like structure with truncated
s-t-graphs as nodes is obtained by recursively choosing a face $f$
wisely, and then letting vertices that can reach vertices on $f$
belong to this node, and partitioning all other vertices among the
descendants of this node.  As we shall see in
Section~\ref{sec:st-tree-construction}, this can be done in such a way
that we obtain logarithmic height and such that the border between a
node and its ancestors is a cycle of alternation number at most $4$.  We
call this the \emph{frame} of the node.

We always choose the truncated s-t-graph maximally, such that once a path crosses
a frame, it does not exit the frame again. Thus, for $u$ to reach $v$,
$u$ has to lie in a component which is ancestral to that of $v$, and
since the alternation number of any frame between those two component
is at most $4$, the path could always be chosen to use one of the at
most $4$ different ``best'' vertices for reaching $v$ on that
frame. Thus, the idea is to do something inspired by level ancestry to
find those ``best'' vertices in $u$'s component.
We handle the case of frames with alternation
number $2$ in Section~\ref{sec:2-frames}. Frames with alternation
number $4$ are similar but more involved, and the details are found in
Section~\ref{sec:4-frames}.
\begin{wrapfigure}{r}{0.3\textwidth}
\includegraphics[width=0.9\linewidth]{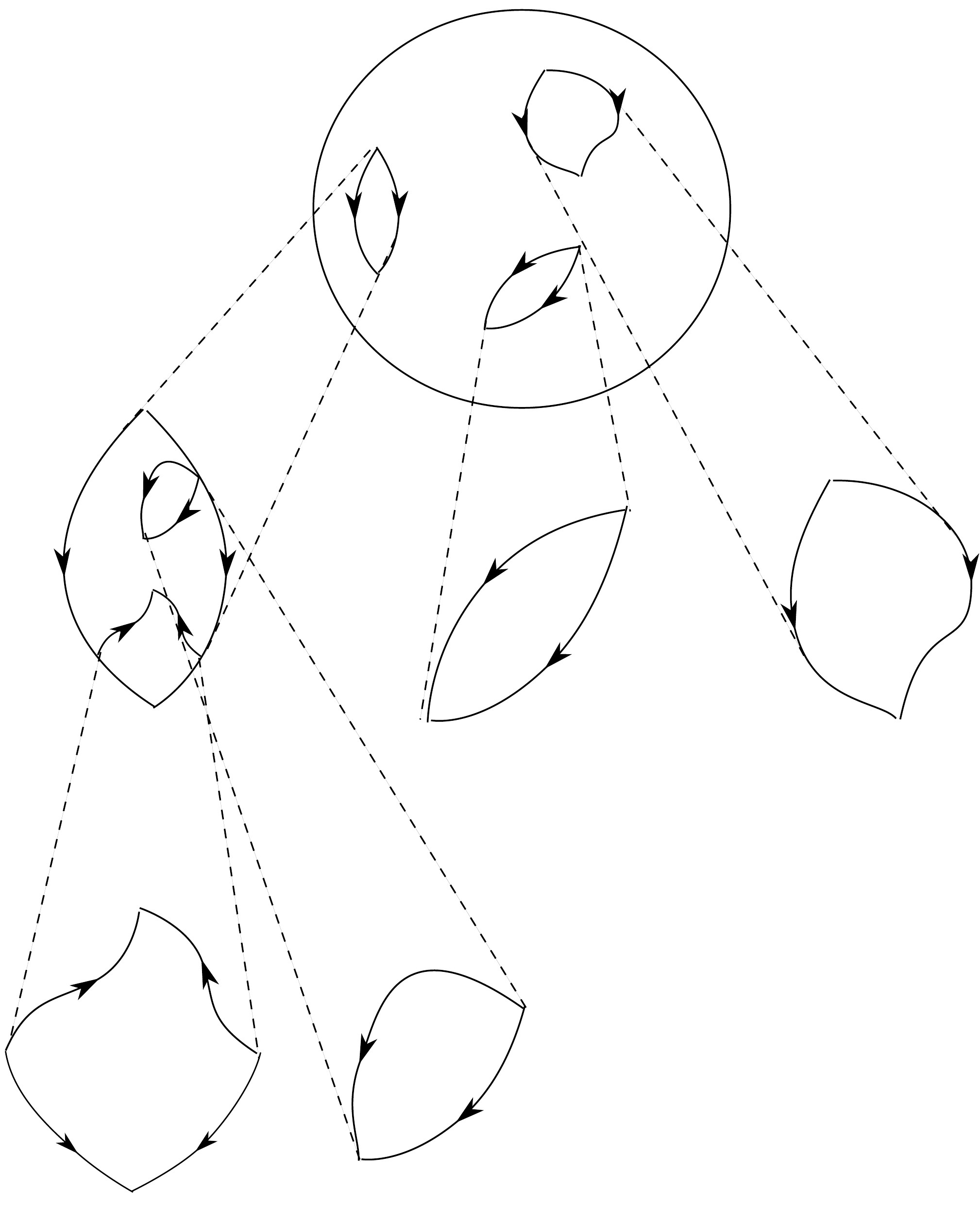}
\caption{A tree of truncated s-t-graphs, each child contained in a face-cycle of its parent.}
\vspace{-40pt}
\label{fig:s-t-decomp}
\end{wrapfigure}

\begin{definition}
  Given a graph $G=(V,E)$, a subgraph $G'=(V',E')$ is
  \emph{backward closed} if
  $\forall{}(u,v)\in{}E:v\in{}V'\implies{}(u,v)\in{}E'$.
\end{definition}
\begin{definition}
  The \emph{backward closure} of a face $f$, denoted
  $\operatorname{bc}(f)$ is the unique smallest backward closed
  graph that contains all the vertices incident to $f$.
\end{definition}

\begin{definition}
  Let $G=(V,E)$ be an acyclic single-source plane digraph, and let
  $\dual{G}=(\dual{V}, \dual{E})$ be its dual.  An
  \emph{s-t-decomposition} of $G$ is a rooted tree where each node
  $x$ is associated with a face $f_x\in\dual{V}$ and subgraphs
  $\dual{G}_x\subseteq\dual{G}$ and $C_x\subseteq{}S_x\subseteq{}G$
  such that:
  \begin{itemize}

  \item $f_x$ is unique ($f_x\neq{}f_y$ for $x\neq{}y$).

  \item $S_x$ is $\operatorname{bc}(f_x)$ if x is the root, and
    $\operatorname{bc}(f_x)\cup{}S_y$ if $x$ is a child of $y$.

  \item $C_x$ is $\operatorname{bc}(f_x)$ if x is the root, and
    $\operatorname{bc}(f_x)\setminus{}S_y$ if $x$ is a child of $y$.

  \item $\dual{G}_x$ is the subgraph of $\dual{G}$ induced by
    $\set{f_z\middle|\:z\text{ is a descendent of }x}$.
    Furthermore, if $x$ is a child of $y$ we require that 
    $\dual{G}_x$ is the connected component of
    $\dual{G}\setminus{}\dual{E}[S_y]$ containing $f_x$.

  \end{itemize}
  If $x$ is a child of $y$, 
  $x$ has a \emph{parent frame}
  $F_x\subseteq{}S_y$ and a set of \emph{down-edges} $E_x\subseteq{}E$ such that:
  \begin{itemize}
  \item $F_x$ is the face cycle in $S_y$ that corresponds to $\dual{G}_x$.
  \item $E_x$ is the set of edges $(w,w')$ such that $w\in{}V[F_x]$ and $w'\in{}V[C_z]$ for some 
    descendant $z$ of $x$.  
  \end{itemize}
  An s-t-decomposition is \emph{good} if the tree has height $\Oo(\log
  n)$ and each frame has alternation number $2$ or $4$.
\end{definition}
\noindent{}The name s-t-decomposition is chosen based on the following
\begin{lemma}
  Each vertex of $G$ is in exactly one $C_x$, and each $C_x$ is a
  truncated s-t-graph.
\end{lemma}
\begin{proof}
  If $x$ is the root, $C_x=\operatorname{bc}(f_x)$ and this is clearly
  a truncted s-t-graph.  Otherwise let $y$ be the parent of $x$.  Then
  $S_x=\operatorname{bc}(f_x)\cup{}S_y$, is backward-closed and
  therefore contains $s$.  Contracting $S_y$ in that graph to a single
  vertex $s'$ gives a single-source graph $S_x/S_y$ with $s'$ as the
  source.  Adding a dummy target $t'$ in $f_x$ results in 
  an s-t-graph $(S_x/S_y)\cup\set{t'}$.  Thus, $S_x/S_y$ is a truncated s-t-graph, and since
  $C_x=\operatorname{bc}(f_x)\setminus{}S_y=S_x\setminus{}S_y=(S_x/S_y)\setminus\set{s'}$
  so is $C_x$.

  Let $v$ be a vertex, let $I$ be the set of all nodes in the s-t-decomposition whose associated faces $\set{f_x}_{x\in I}$ are reachable from $v$, and let $N=\operatorname{lca}(I)$. We now show that $v$ lies in $C_N$ and only in $C_N$.
  To see that $v\in C_N$, note that $v\in S_x$ for all $x\in I$, but then $v\in \bigcap_{x\in I} S_x = S_N$. But $v\notin S_a$ for any ancestor $a$ of $N$ by definition of $\operatorname{lca}$, and thus, $v\notin S_y$ for the parent $y$ of $N$, entailing $v\in S_N\setminus S_y = C_N$. We have now seen that $v\in C_N$ and that $v\notin C_x$ when $x\prec N$ or $N\prec x$. To see that $v\notin C_x$ for any unrelated $x\neq N$, note the following: if $x$ has no descendants in $I$, then $v\not\in{}V[S_x]$ since all vertices reachable from $v$ lie on some face. Thus, $v\notin C_x\subseteq S_x$.


\end{proof} 

\begin{theorem}\label{thm:st-tree-exists}
  Any acyclic single-source plane digraph has a good s-t-decomposition.
\end{theorem}
We defer the proof to section~\ref{sec:st-tree-construction}.  The
reason for studying s-t-decompositions in the context of reachability is
the following
\begin{lemma}
  If $\reachable{u}{v}$ where $u\in{}C_x$ and $v\in{}C_y$ then either
  $x=y$ or $x$ has a child $z$ that is ancestor to $y$ such that any
  $\simplepath{u}{v}$ path contains a vertex in $F_z$.
\end{lemma}
\begin{proof}
  Note that whenever $\reachable{w}{w'}$ with $w'\in C_a$, $w$ must belong to an ancestor of $a$, since $w\in bc(f_a)$. Thus, $x$ is an ancestor of $y$, which means that either $x=y$ or $x$ has a child $z$ that is an ancestor to $y$. But then either $w$ lies on $F_x$, or $F_x$ is a cycle separating $w$ from $w'$. In either case, a path from $w$ to $w'$ must contain a vertex on $F_x$. 
\end{proof}
Since (by theorem~\ref{thm:st-tree-exists}) we can assume the
alternation number is at most $4$, this reduces the reachability
question to the problem of finding the at most $4$ ``last'' vertices
on $F_z\cap{}C_x$ that can reach $v$ and then checking in $C_x$ if $u$
can reach either of them.  In section~\ref{sec:2-frames} we will show
how to do this efficiently when $F_z$ is \emph{a $2$-frame}, that is, has alternation number $2$, and
in section~\ref{sec:4-frames} we will extend this to the case when
$F_z$ is \emph{a $4$-frame}, that is, has alternation number $4$.

\begin{theorem}
There exists a practical RAM data structure that for any planar digraph with $n$ vertices uses $O(n)$ words of $O(\log n)$ bits and can answer reachability queries in constant time. The data structure can be built in linear time.
\end{theorem}
\begin{proof}
First, build a good s-t-decomposition of $G$. Such a decomposition exists (Lemma \ref{thm:st-tree-exists}) and can be built in linear time (Lemma \ref{lem:st-linear-time}). Adding DFS pre- and postorder numbers to each node in the tree lets us discover the ancestry relationship between any two vertices in constant time. 
Then, calculate the structures described in Section \ref{sec:2-frames} (in particular $d_2[]$) and Section \ref{sec:4-frames} ($c[]$ and $d[]$).

To answer $\operatorname{reachable}(u,v)$, there are the following cases. Let $x=c[u]$ and $y=c[v]$. 
\begin{enumerate}
\item If $x\not\preceq{}y$, then $u$ cannot reach $v$. 
\item If $x=y$, then the answer is given by the s-t-graph labelling of $C_x$ from~\cite{DBLP:journals/tcs/TamassiaT93}. 
\item If $x\prec{}y$ and $d_2[u]==d_2[v]$ there are no $2$-frames separating $u$ and $v$, but since $x\prec y$ there are $4$-frames. Let $i=d[u]$, then by~\ref{thm:4-frames} we can in constant time compute $l^0_i(v)$, $r^0_i(v)$, $l^1_i(v)$, and $r^1_i(v)$.  If $u$ can reach any of them, them $u$ can reach $v$, otherwise no.
\item Otherwise $x\prec{}y$ and $d_2[u]<d_2[v]$ and there is a $2$-frame separating $u$ and $v$.  Let $i=d_2[u]$, then by~\ref{thm:2n-frames} we can in constant time compute $l_i(v)$ and $r_i(v)$. If $u$ can reach any of them, then $u$ can reach $v$, otherwise no.
\end{enumerate}
Note that the recursive calls in step $3$ only leads to questions of type $<3$, and similarly the recursive calls in step $4$ only leads to questions of type $<4$.  Thus any query uses case $3$ at most twice and case $1+2$ at most $8$ times. Thus we use only constant time per query.
\end{proof}

A consequence of our construction which might be of independent interest is the following:

\begin{theorem}
	If a planar digraph $G$ admits an s-t-decompostion of height $h$ where all frames have alternation number $2$ and $4$, there exists an $O(h\log n)$ bit labelling scheme for reachability with evaluation time $O(h)$
\end{theorem}

Especially, if a class of planar digraphs have such an s-t-decompositions of constant height, they have an $O(\log n)$ bit labelling scheme for reachability.

\input{st-tree-construction.tex}
\input{st-tree-linear-time.tex}

\input{2-frame-new.tex}

%% file: st-tree-construction.tex
\subsection{Constructing an s-t-decomposition}\label{sec:st-tree-construction}

The s-t-decomposition recursively chooses a face $f$ and consequently a subgraph $H=bc(f)$ of the graph $G$ induced by all vertices that can reach a vertex on $f$. Since $G$ was embedded in the plane, the subgraph $H$ is embedded in the plane, and all vertices of $G\setminus H$ lie in a unique face of $H$. We may choose a tree-cotree composition wisely, such that for each face of $H$, the restriction of $T^\ast$ to the subfaces of that face is again a dual spanning tree (Lemma~\ref{lem:tstar-connected}).

We also have to choose $H$ carefully to ensure logarithmic height, and a limited alternation number on the frames. To ensure at most logarithmic height, we show two cases: $2$-frame-nodes have only small children, while for $4$-frame-nodes, we only need to ensure that their $4$-frame children themselves are small. 

\begin{lemma}\label{lem:face-subtree}
  Let $G=(V,E)$ be a plane graph, let $\dual{G}=(\dual{V}, \dual{E})$
  be its dual, let $(T,\dual{T})$ be a tree/cotree decomposition of
  $G$, and let $S$ be a subgraph of $G$ such that $S\cap{}T$ is
  connected.  Then the faces of $S$ correspond to connected components
  of $\dual{T}\setminus\dual{E}[S]$.
\end{lemma}
\begin{proof}
  Let $\dual{S}$ be the dual of $S$, then
  $\dual{S}=\dual{G}/(\dual{G}\setminus\dual{E}[S])$ and the claim is
  equivalent to saying that the components of
  $\dual{G}\setminus\dual{E}[S]$ correspond to the components of
  $\dual{T}\setminus\dual{E}[S]$.  Consider a pair of faces
  $f_1,f_2\in\dual{V}$. Clearly, if they are in separate components of
  $\dual{G}\setminus\dual{E}[S]$, they are also in
  separate components in $\dual{T}\setminus\dual{E}[S]$. On the other hand, suppose
  $f_1$ and $f_2$ are in different components in
  $\dual{T}\setminus\dual{E}[S]$.  Then there exists an edge
  $\dual{e}\in\dual{E}[S]\cap\dual{T}$ separating them.  The
  corresponding edge $e\in{}E[S]$ induces a cycle in $T$, which is
  also part of $S$ since $S\cap{}T$ is connected.  The dual to that
  cycle is an edge cut in $\dual{G}$ that separates $f_1$ from $f_2$.
\end{proof}

\begin{lemma}\label{lem:tstar-connected}
  Let $T$ be a spanning tree where all edges point away from the source $s$ of $G$,
  then for any node $x$ in an st-decomposition of $G$, the subgraph
  $\dual{T}_x$ of $\dual{T}$ induced by $\dual{V}[\dual{G}_x]$
  is a connected subtree of $\dual{T}$.
\end{lemma}
\begin{proof}
  If $x$ is the root, 
  this trivially holds. If $x$ has a parent $y$,
  $\dual{G}_x$ corresponds to a face in $S_y$.  Now $S_y\cap{}T$ is
  connected since $S_y$ is the union of backward-closed graphs, and
  the result follows from Lemma~\ref{lem:face-subtree}.
\end{proof}

\begin{lemma}\label{lem:2-face-split}
  Let $x$ be a node in an st-decomposition whose parent frame $F_x$
  has alternation number $2$, and let $\dual{A}$ be the set of faces
  in $\dual{T}_x$ incident to the target corner of $F_x$.  Then for
  any child $y$ of $x$:
  \begin{alignat*}{3}
    \dual{A}&\subseteq\dual{V}[\dual{T}_y]
    &&\quad\implies\quad&
    &F_y\text{ has alternation number }4.
    \\
    \dual{A}&\not\subseteq\dual{V}[\dual{T}_y]
    &&\quad\implies\quad&
    &F_y\text{ has alternation number }2.
  \end{alignat*}
\end{lemma}
\begin{proof}
  Let $t_x$ be the target corner of $F_x$ and let $\dual{A}$ be the
  set of faces in $\dual{T_x}$ incident to $t_x$.
  For any child $y$ if $x$, $F_y$ consists of a (possibly empty)
  segment of $F_x$ and two directed paths that meet at a new target
  corner $t_y$.  Each target corner of $F_y$ must therefore be at
  either $t_x$ or $t_y$.
  Now if $\dual{A}\subseteq\dual{V}[\dual{T}_y]$, then both $t_x$ and
  $t_y$ are target corners of $F_y$, otherwise only $t_y$ is.  Either
  way the result follows.
\end{proof}

\begin{lemma}\label{lem:4-face-split}
  Let $x$ be a node in an st-decomposition whose parent frame $F_x$
  has alternation number $4$, and let $\dual{A^0}$ and $\dual{A^1}$ be
  the sets of faces in $\dual{T}_x$ incident to the target corners of
  $F_x$.  Then for any child $y$ of $x$:
  \begin{alignat*}{3}
    \dual{A^0}\not\subseteq\dual{V}[\dual{T}_y]&\vee{}\dual{A^1}\not\subseteq\dual{V}[\dual{T}_y]
    &&\quad\implies\quad&
    &F_y\text{ has alternation number at most }4.
    \\
    \dual{A^0}\not\subseteq\dual{V}[\dual{T}_y]&\wedge{}\dual{A^1}\not\subseteq\dual{V}[\dual{T}_y]
    &&\quad\implies\quad&
    &F_y\text{ has alternation number }2.
  \end{alignat*}
\end{lemma}
\begin{proof}
  Let $t^0_x$ and $t^1_x$ be the two target corners of $F_x$ and for
  $i\in\set{0,1}$ let $\dual{A^i}$ be the set of faces in $\dual{T_x}$
  incident to $t^i_x$.
  For any child $y$ of $x$, $F_y$ consists of a (possibly empty)
  segment of $F_x$ and two directed paths that meet at a new target
  corner $t_y$.  Each target corner of $F_y$ must therefore be at
  either $t_y$, $t^0_x$, or $t^1_x$.
  Now if $\dual{A^i}\not\subseteq\dual{V}[\dual{T}_y]$ for some
  $i\in\set{0,1}$, then $t^i_x$ is not a target corner of $F_y$.  So
  the number of target corners in $F_y$ is at least $1$, and at most
  $3$ minus the number of such $i$, and the result follows.
\end{proof}

\begin{proof}[proof of theorem~\ref{thm:st-tree-exists}]
  Let $s$ be the source of $G$ and let $(T,\dual{T})$ be a tree/cotree
  decomposition of $G$ such that all edges in $T$ point away from $s$.
  The st-decomposition can be constructed recursively as follows.
  Start with the root.  In each step we have a node $x$ and by
  Lemma~\ref{lem:tstar-connected} the subgraph $\dual{T}_x$ induced in
  $\dual{T}$ by $\dual{V}[\dual{G}_x]$ is a tree.  The goal is to
  select a face $f_x$ such that for each child $y$:
  \begin{itemize}
  \item The alternation number of $F_y$ is at most $4$, and
  \item For each child $z$ of $y$ (and thus grandchild of $x$),
    $|\dual{T}_z|\leq\frac{1}{2}|\dual{T}_x|$.
  \end{itemize}
  If we can do this for all $x$, we are done.  There are $3$ cases:
  \paragraph{$x$ is the root} Let $f_x$ be the median of
  $\dual{T}_x=\dual{T}$.  Then for each child $y$,
  $|\dual{T}_y|\leq\frac{1}{2}|\dual{T}_x|$, and, since $S_x=\operatorname{bc}(f_x)$ is a
    truncated s-t-graph with a single source, $f_y$ has alternation number $2$.

  \paragraph{$F_x$ has alternation number $2$} 
  Let $f_x$ be the median of $\dual{T}_x$. Then for each child $y$,
  $|\dual{T}_y|\leq\frac{1}{2}|\dual{T}_x|$, and, by
  Lemma~\ref{lem:2-face-split}, $f_y$ has alternation number at most
  $4$.  

  \paragraph{$F_x$ has alternation number $4$} 
  Let $t_0$ and $t_1$ be the
    local targets of $F_x$ and let $f_0,f_1\in\dual{V}[\dual{T}_x]$ be
    (not necessarily distinct) faces incident to $t_0$ and $t_1$
    respectively.  Now choose $f_x$ as the projection of the median
    $m$ of $\dual{T}_x$ on the path $f_0,\ldots,f_1$ in $\dual{T}_x$.
    By Lemma~\ref{lem:4-face-split} this means that for any child $y$
    of $x$, the alternation number of the parent frame $F_y$ is at
    most $4$.
\\- If $f_x=m$ then 
    $|\dual{T}_y|\leq\frac{1}{2}|\dual{T}_x|$.
%
\\- If $f_x\neq{}m$ and $\dual{T}_y$ is not the component
    of $m$ in $\dual{T}_x\setminus\dual{E}[\operatorname{bc}(f_x)]$,
    then 
    $|\dual{T}_y|\leq\frac{1}{2}|\dual{T}_x|$.
%
\\- If $f_x\neq{}m$, and $\dual{T}_y$ is the component of $m$ in
    $\dual{T}_x\setminus\dual{E}[\operatorname{bc}(f_x)]$, then
    $\dual{T}_y$ contains neither $f_0$ nor $f_1$, so by
    Lemma~\ref{lem:4-face-split} the parent frame $F_y$ has
    alternation number at most $2$ and we have just shown this means
    any child $z$ of $y$ has
    $|\dual{T}_z|\leq\frac{1}{2}|\dual{T}_y|\leq\frac{1}{2}|\dual{T}_x|$.
\end{proof}

Note that this construction can be implemented in linear time by using
ideas similar to~\cite{Alstrup97optimalon-line}.

%% file: st-tree-linear-time.tex
\subsection{Constructing a good s-t-decomposition in linear time}
In the construction of an s-t-decomposition, a face is chosen, some edges are deleted, and new connected components of the dual graph arise. We then recurse on the new connected components of the dual graph. 
By Lemma~\ref{lem:tstar-connected} we can choose a tree/cotree-decomposition such that each component that arises is spanned by a subtree of the cotree.

To obtain linear construction time, we use a variation of the decremental tree connectivity algorithm from~\cite{Alstrup97optimalon-line} to keep track of the subtrees of the cotree, and associate some
information with each subtree. In particular, when $\dual{T}_x$ is a component at some point, we can in constant time find the node $x$.

For each node $x$ we keep the
set of target vertices on $F_x$ (or $\emptyset$ if $x$ is the root),
and a face in $\dual{T}_x$ incident to each target in the set.

Build a top tree (see \cite{Alstrup:2005}) of height $O(\log n)$ over $\dual{T}$, 
and let $\dual{v}_{n-i}$ be the $i$'th face that stops being boundary 
during the construction. Using this enumeration, the boundary faces of a cluster will be visited before boundary faces of their descendants. We use this ordering to find the splitting faces of the s-t-decomposition.

For each $\dual{v}_i$, we can use the connectivity structure to find
the relevant node $x$ to split.  We then need to choose the target face $f_x$ defining the split.  
If $x$ is the root or $F_x$ is a $2$-frame, we just set $f_x=\dual{v}_i$.
If $F_x$ is a $4$-frame, the
information in $x$ contains a pair of faces $f_1,f_2$ and we use
a static nearest common ancestor data structure from Harel and
Tarjan~\cite{DBLP:journals/siamcomp/HarelT84} to find the projection $f_x=\pi(\dual{v}_i)$ of $\dual{v}_i$ on
$f_1,\cdots,f_2$. Note that the projection of $\dual{v}_i$ is always contained in the same connected component as $f_1,f_2$, and thus, the data structure for the whole tree suffices to answer this query for the particular subtree.

Once $f_x$ has been selected, we traverse the graph backwards from the
vertices of $f_x$ until we have found all the edges with destination in
$C_x$.  This search takes $\sizeof{C_x}$ time.  We delete these edges
from the forest as we go along.  Once we are done, we take all targets
in $C_x$ and select an incident face for each component it is incident
to. This again takes $\sizeof{C_x}$ time.  If $f_x\neq\dual{v}_i$ we try with $\dual{v}_i$ again, otherwise we move on to $\dual{v}_{i+1}$.

\begin{lemma}\label{lem:st-alternation-number}
	The s-t-decomposition constructed via the approach sketched above has no frame of alternation number $>4$.
\end{lemma}

\begin{proof}
	Components with $2$-frames always have children with $2$- and $4$-frames. For components with $4$-frames, this follows directly from Lemma \ref{lem:2-face-split}, since we chose a splitting face on the cotree path between faces near the two targets.
\end{proof}

\begin{lemma}\label{lem:st-height}
	The s-t-decomposition constructed via the approach sketched above has height $O(\log n)$.
\end{lemma}

\begin{proof}
	Since the top-tree has height $O(\log n)$, choosing the boundary face $\dual{v}_i$ as a splitting face every time would result in a tree of the same height; $O(\log n)$. However, for each $4$-frame, we might choose a face $f_x\neq\dual{v}_i$ which is the projection of $\dual{v}_i$ on $f_1 \ldots f_2$.  As noted in Lemma~\ref{lem:2-face-split}, when this happens, $\dual{v}_i$ will lie in a child which has a $2$-frame. But then, $\dual{v}_i$ will be the splitting face for that child. We thus increase the height by no more than a factor $2$, and the s-t-decomposition has height $2 O(\log n) = O(\log n)$.
\end{proof}


\begin{lemma}\label{lem:st-linear-time}
	Let $G=(V,E)$ be a plane single-source graph with source $s$, then we can construct a good s-t-decomposition of $G$ in linear time.
\end{lemma}

\begin{proof}
	Since the top-tree can be constructed in linear time, and since the decremental connectivity for trees takes linear time, and since the static nearest common ancestor data stucture is constructed in linear time and answers queries in constant time, the construction takes linear time.
	By Lemma \ref{lem:st-alternation-number} and \ref{lem:st-height}, the resulting s-t-decomposition is good.
\end{proof}

%
%

%% file: 2-frame-new.tex
\subsection{2-frames}\label{sec:2-frames}

\begin{definition}
  Let $\mathcal{T}$
  be an st-decomposition of $G=(V,E)$.  Then we can define a
  \emph{$2$-frame-decomposition} $\mathcal{T}_2$ by contracting each
  edge in $\mathcal{T}$ that corresponds to a $4$-frame.  For each
  node $x$ in $\mathcal{T}_2$ that is contracted from a set of nodes
  $Y\subseteq\mathcal{T}$ define $C_x:=\bigcup_{y\in{}Y}C_y$ and if
  $x$ is not the root, define $F_x:=F_{\lca(Y)}$ and
  $E_x:=E_{\lca(Y)}$.  Then $F_x$ is a $2$-frame, and we can define
  $s_x$ to be the source corner, and $t_x$ to be the target corner on
  $F_x$.
\end{definition}


\begin{definition}\label{def:2n-global-LR-partition}
  Let $(\mathcal{L},\mathcal{R})$ be the partition of
  $\cup_{x\in\mathcal{T}_2}E_x$ defined as follows: For each
  $(u,v)\in\cup_{x\in\mathcal{T}_2}E_x$ let $y$ be the node (if it
  exists) closest to the root of $\mathcal{T}_2$ such that
  $(u,v)\in{}E_y$ but $u$ is not the target vertex of $F_y$.  If $y$
  exists and $(u,v)$ is incident to a corner on the clockwise
  disegment of $F_y$ between $s_y$ and $t_y$ assign $(u,v)$ to
  $\mathcal{R}$, otherwise assign $(u,v)$ to $\mathcal{L}$.
\end{definition}





\begin{definition}
  Let $\mathcal{T}_2$ be an $2$-frame-decomposition of $G=(V,E)$.  For any
  vertex $v\in{}V$ define:
  \begin{align*}
    c_2[v]&:=\text{The node }x\text{ in }\mathcal{T}_2\text{ such that }v\in{}V[C_x]
    \\
    d_2[v]&:=\text{The depth of }c_2[v]\text{ in }\mathcal{T}_2
  \end{align*}
\end{definition} 
  
\begin{figure}
\centering
\includegraphics[width=0.4\linewidth]{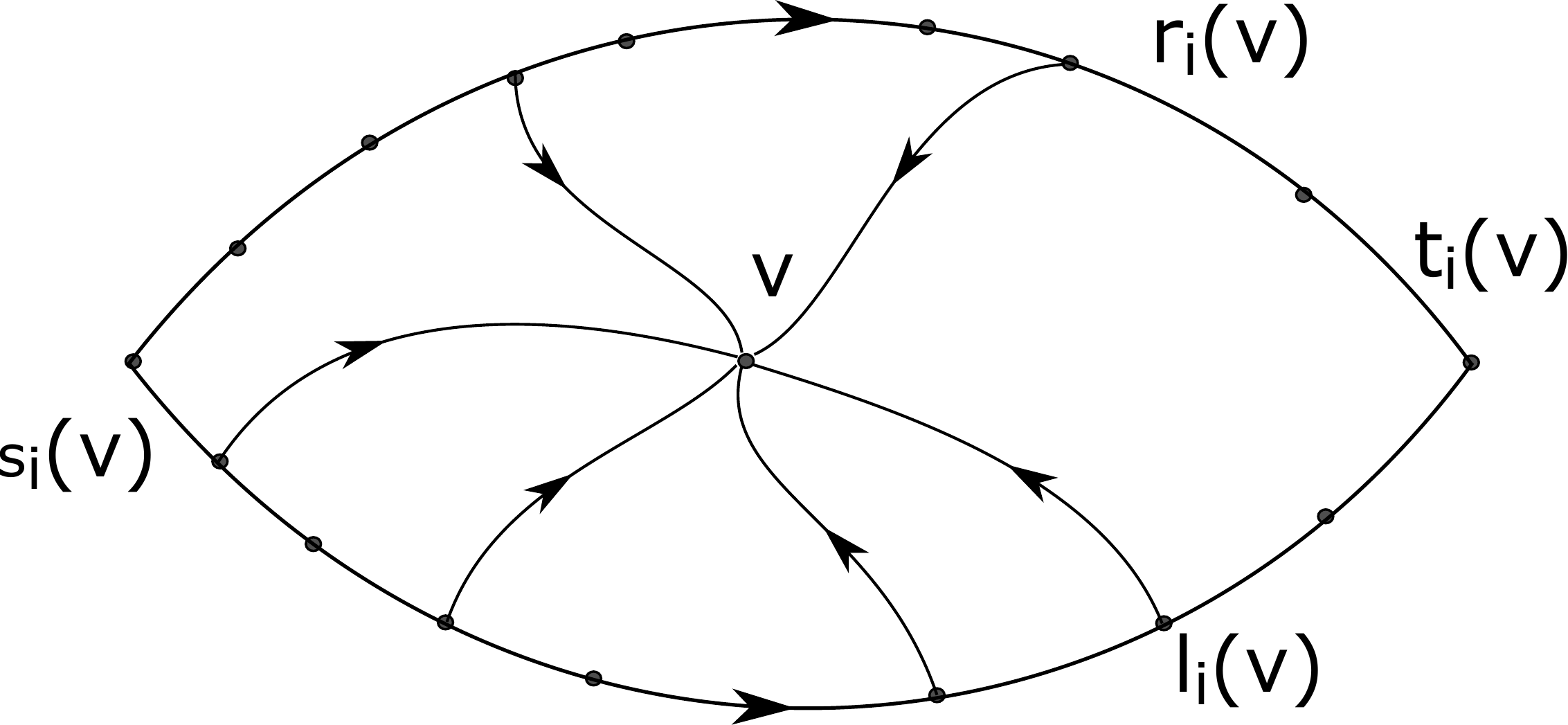}
\caption{The best two vertices that can reach $v$ on level $i$.}
\label{fig:l_i-and-r_i}
\end{figure}

\begin{definition}\label{def:2n-i-of-v}
  For any $0\leq{}i<d_2[v]$, let $x$ be the
  ancestor of $c_2[v]$ at depth $i+1$ and define:
  \begin{align*}
    E_i(v)&:=E_x
    \\
    L_i(v)&:=E_x\cap\mathcal{L}
    \\
    R_i(v)&:=E_x\cap\mathcal{R}
    \\
    \widehat{L}_i(v)&:=\set{(w,w')\in{}L_i(v)\middle|\:
      \reachable{w'}{v}
    }
    \\
    \widehat{R}_i(v)&:=\set{(w,w')\in{}R_i(v)\middle|\:
      \reachable{w'}{v}
    }
    \\
    \widehat{F}_i(v)&:=\widehat{L}_i(v)\cup\widehat{R}_i(v)
    \\
    l_i(v)&:=
    \begin{cases}
      \bot&\text{ if }\widehat{L}_i(v)=\emptyset
      \\
      \text{the last vertex in $\source(\widehat{L}_i(v))$ on the counterclockwise dipath of $F_x$}&\text{ otherwise}
    \end{cases}
    \\
    r_i(v)&:=
    \begin{cases}
      \bot&\text{ if }\widehat{R}_i(v)=\emptyset
      \\
      \text{the last vertex in $\source(\widehat{R}_i(v))$ on the clockwise dipath of $F_x$}&\text{ otherwise}
    \end{cases}
    \\
    s_i(v)&:=\text{The vertex associated with }s_x
    \\
    t_i(v)&:=\text{The vertex associated with }t_x
  \end{align*}
  Additionally, let $L_i(v)$ and $\widehat{L}_i(v)$ be totally ordered
  by the position of the starting vertices on the counterclockwise
  disegment of $F_x$ and the clockwise order around each starting
  vertex.  Similarly let $R_i(v)$ and $\widehat{R}_i(v)$ be totally
  ordered by the position of the starting vertices on the
  clockwise disegment of $F_x$ and the counterclockwise order around
  each starting vertex.
\end{definition}

The goal in this section is a data structure for efficiently computing
$l_i(v)$ and $r_i(v)$ for $0\leq{}i<d_2[v]$.

\begin{lemma}\label{lem:2n-hat-nonempty}
  For any vertex $v\in{}V$ and $0\leq{}i<d_2[v]$:
  $\widehat{F}_i(v)\neq\emptyset$
\end{lemma}
\begin{proof}
  Let $x$ be the ancestor of $c_2[v]$ at depth $i+1$.  Since $G$ is a
  single-source graph, there is a path from $s$ to $v$.  This path
  must contain a vertex in $V[F_x]$. 
  But then the edge following the last such vertex on the path must be
  in $\widehat{L}_i(v)\cup{}\widehat{R}_i(v)$ which is therefore
  nonempty.
\end{proof}

\begin{lemma}\label{lem:2n-hat-subset}
  For any $u,v\in{}V$ and $0\leq{}i<d_2[u]$: If $\reachable{u}{v}$ then
  $\widehat{L}_i(u)\subseteq\widehat{L}_i(v)$ and $\widehat{R}_i(u)\subseteq\widehat{R}_i(v)$.
\end{lemma}
\begin{proof}
  Since $\reachable{u}{v}$, $c_2[u]$ is ancestor to $c_2[v]$ and so
  $L_i(u)=L_i(v)$ and hence
  $\widehat{L}_i(u)\subseteq\widehat{L}_i(v)$.  Similarly,
  $R_i(u)=R_i(v)$ and $\widehat{R}_i(u)\subseteq\widehat{R}_i(v)$.
\end{proof}

\begin{lemma}\label{lem:2n-lr-multiframe}
  Given any vertex $v\in{}V$, $0\leq{}i<d_2[v]$, and $(w,w')\in{}E_i(v)$. Then:
  \begin{align*}
    (w,w')&\in\widehat{L}_i(v)
    &&\implies&
    (w,w')&\in\widehat{L}_{i'}(v)\text{ for all }i^\prime, d_2[w]\le i^\prime <\min\set{d_2[w'],d_2[v]}
    \\
    (w,w')&\in\widehat{R}_i(v)
    &&\implies&
    (w,w')&\in\widehat{R}_{i'}(v)\text{ for all }i^\prime, d_2[w]\le i^\prime <\min\set{d_2[w'],d_2[v]}
  \end{align*}
\end{lemma}
\begin{proof}
  Let $j=d_2[w]$ and $k=\min\set{d_2[w'],d_2[v]}$.
  Clearly $(w,w')\in{}E_{i'}$ for all $j\leq{}i'<k$.  Suppose
  $(w,w')\in\widehat{L}_i(v)\subseteq{}L_i(v)$, then
  since $j\leq{}i<k$ the definition give us
  $(w,w')\in{}L_{i'}(v)$ for all $j\leq{}i'<k$.  And since
  $\reachable{w'}{v}$ this implies
  $(w,w')\in\widehat{L}_{i'}(v)$ for all $j\leq{}i'<k$ and the result follows.
  The case for $R$ is symmetric.
\end{proof}

\begin{definition}\label{def:2n-pl-pr}
  For any vertex $v\in{}V$ let
  \begin{align*}
    p_l[v]&:=
    \begin{cases}
      \bot&\text{ if }d_2[v]=0
      \\
      l_{d_2[v]-1}(v)&\text{ otherwise}
    \end{cases}
    \\
    p_r[v]&:=
    \begin{cases}
      \bot&\text{ if }d_2[v]=0
      \\
      r_{d_2[v]-1}(v)&\text{ otherwise}
    \end{cases}
  \end{align*}
  and let $T_l$ and $T_r$ denote the rooted forests over $V$ whose
  parent pointers are $p_l$ and $p_r$ respectively.
\end{definition}

\begin{definition}\label{def:2n-lr-prime}
  For any $v\in{}V\cup\set{\bot}$, and $i\geq{}0$ let
  \begin{align*}
    l'_i(v)&:=
    \begin{cases}
      v&\text{ if $v=\bot\vee{}d_2[v]\leq{}i$}
      \\
      l'_i(p_l[v])&\text{ otherwise}
    \end{cases}
    \\
    r'_i(v)&:=
    \begin{cases}
      v&\text{ if $v=\bot\vee{}d_2[v]\leq{}i$}
      \\
      r'_i(p_r[v])&\text{ otherwise}
    \end{cases}
  \end{align*}
\end{definition}

\begin{lemma}\label{lem:2n-lr-prime-basics}
  Let $v\in{}V$, and $i\geq{}0$ be given, then
  \begin{align*}
    i&=d_2[v]-1
    &&\implies&
    l'_i(v)&=l_i(v)
    &&\wedge&
    r'_i(v)&=r_i(v)
    \\
    i&\leq{}d_2[v]-1
    &&\implies&
    l'_i(v)&\in\source(\widehat{L}_i(v))\cup\set{\bot}
    &&\wedge&
    r'_i(v)&\in\source(\widehat{R}_i(v))\cup\set{\bot}
    \\
    i&>d_2[v]-1
    &&\implies&
    l'_i(v)&=v
    &&\wedge&
    r'_i(v)&=v
  \end{align*}
\end{lemma}
\begin{proof}
  We will show this for $l'$ only, as $r'$ is completely symmetrical.
  If $i>d_2[v]-1$ then $d_2[v]\leq{}i$ and we get $l'_i(v)=v$ directly
  from the definition of $l'$.
  Similarly if $i=d_2[v]-1$ then $l'_i(v) =
  l'_i(p_l[v]) =
  l'_i(l_{d_2[v]-1}(v)) =
  l'_i(l_i(v)) =
  l_i(v) \in
  \source(\widehat{L}_i(v))\cup\set{\bot}$.
  Finally suppose $i<d_2[v]-1$.  If $l'_i(v)=\bot$ we are done,
  so suppose that is not the case.  Let $u$ be the child of
  $l'_i(v)$ in $T_l$ that is ancestor to $v$.  Then
  $l'_i(v) = l'_i(u) = p_l[u] =
  l_{d_2[u]-1}(u)$.  By definition of $l_{d_2[u]-1}(u)$ there
  exists an edge $(w,w')\in\widehat{L}_{d_2[u]-1}$ where
  $w=l_{d_2[u]-1}(u)$ and $d_2[w]\leq{}i<d_2[w']\leq{}d_2[u]$ and by
  setting $(v,i,(w,w')) = (u,d_2[u]-1,(w,w'))$ in
  lemma~\ref{lem:2n-lr-multiframe} we get
  $(w,w')\in\widehat{L}_i(u)$, and therefore
  $l'_i(v)\in\source(\widehat{L}_i(u))$.  But since
  $\reachable{u}{v}$ we have
  $\widehat{L}_i(u)\subseteq\widehat{L}_i(v)$ by Lemma~\ref{lem:2n-hat-subset} and we are
  done.
\end{proof}

\begin{lemma}\label{lem:2n-lr-prime-reduce}
  Let $v\in{}V$ and $0\leq{}i\leq{}j$ then
  \begin{align*}
    l'_i(l'_j(v))&=l'_i(v)
    &&\wedge&
    r'_i(r'_j(v))&=r'_i(v)
  \end{align*}
\end{lemma}
\begin{proof}
  $l'_j(v)$ is on the path from $v$ to $l'_i(v)$ in $T_l$, so this
  follows trivially from the recursion.  The case for $r'$ is symmetric.
\end{proof}

\begin{lemma}\label{lem:2n-lr-bot}
  Let $v\in{}V$, and $0\leq{}i<d_2[v]-1$, then
  \begin{alignat*}{7}
    l_i(v)&=\bot
    &&\quad\implies\quad&
    l'_i(l_{i+1}(v))&=\bot
    &&\qquad\wedge\qquad&
    r_i(v)&=\bot
    &&\quad\implies\quad&
    r'_i(r_{i+1}(v))&=\bot
  \end{alignat*}
\end{lemma}
\begin{proof}
  If $l_i(v)=\bot$ then $\widehat{L}_i(v)=\emptyset$, so
  either $l_{i+1}(v)=\bot$ implying
  $l'_i(l_{i+1}(v))=\bot$ by the definition of $l'$, or
  $l_{i+1}(v)\not\in\source(\widehat{L}_i(v))$ so
  $d_2[l_{i+1}(v)]=i+1$ and by Lemma~\ref{lem:2n-lr-prime-basics} and Lemma~\ref{lem:2n-hat-subset}
  $l'_i(l_{i+1}(v)) \in
  \source(\widehat{L}_i(l_{i+1}(v)))\cup\set{\bot}
  \subseteq \source(\widehat{L}_i(v))\cup\set{\bot}=\set{\bot}$
  so again $l'_i(l_{i+1}(v))=\bot$.  The case for $r$ is
  symmetric.
\end{proof}

\begin{lemma}[Crossing lemma]\label{lem:2n-crossing}
  Let $v\in{}V$, and $0\leq{}i<d_2[v]-1$.
  \begin{alignat*}{7}
    l_i(v)&\neq{}l'_i(l_{i+1}(v))
    &&\implies\quad&
    l_i(v)&=l'_i(m)
    &&\:\wedge\:&
    r_i(v)&=r'_i(m)
    &&\:\wedge\:&
    d_2[m]&=i+1
    \\&&&\text{where }m=r_{i+1}(v)\neq\bot\hspace{-30cm}
    \\
    r_i(v)&\neq{}r'_i(r_{i+1}(v))
    &&\implies\quad&
    l_i(v)&=l'_i(m)
    &&\:\wedge\:&
    r_i(v)&=r'_i(m)
    &&\:\wedge\:&
    d_2[m]&=i+1
    \\&&&\text{where }m=l_{i+1}(v)\neq\bot\hspace{-30cm}
  \end{alignat*}
\end{lemma}

\begin{SCfigure}
\centering
\includegraphics[width=0.35\linewidth]{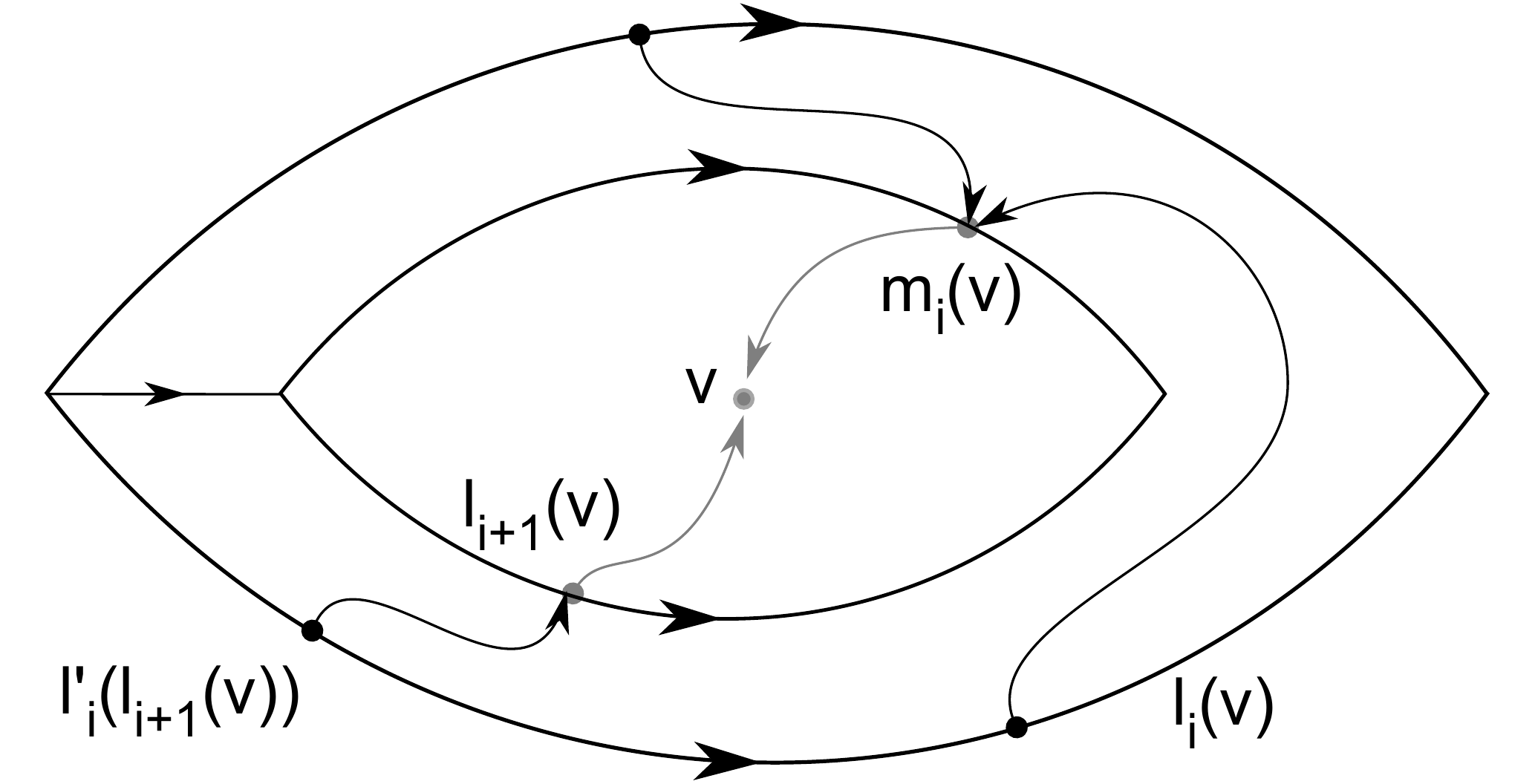}
\caption{The best path from $L_i(v)$ goes via $r_{i+1}(v)$.}
\label{fig:crossing-a}
\end{SCfigure}

\begin{proof}
  Suppose $l_i(v)\neq{}l'_i(l_{i+1}(v))$ (the case
  $r_i(v)\neq{}r'_i(r_{i+1}(v))$ is symmetrical).
  Then $l_i(v)\neq\bot$ by lemma~\ref{lem:2n-lr-bot}.
  Thus there is a last edge $(w,w')\in\widehat{L}_i(v)$ with
  $w=l_i(v)$ and $d_2[w]\leq{}i<d_2[w']$ and a path
  $P=\simplepath{w'}{v}$.

  Now $(w,w')\not\in{}E_{i+1}(v)$
  since otherwise by Definition~\ref{def:2n-i-of-v}
  $(w,w')\in{}L_{i+1}(v)$ and since $\reachable{w'}{v}$ even
  $(w,w')\in\widehat{L}_{i+1}(v)$ implying
  $l_i(v)=l_{i+1}(v)$ and thus
  $l_i(v)=l'_i(l_{i+1}(v))$ by
  lemma~\ref{lem:2n-lr-prime-basics}, contradicting our assumption.

  Since $(w,w')\not\in{}E_{i+1}(v)$, the path $P$ must cross
  $\widehat{F}_{i+1}(v)$.  Let $(u,u')$ be the last edge in
  $P\cap{}\widehat{F}_{i+1}(v)$.  Then $\reachable{w'}{u}$ so
  $d_2[u]\geq{}i+1$ and $(u,u')\not\in{}L_{i+1}(v)$ since
  otherwise $d_2[l_{i+1}(v)]=i+1$ and hence by
  Lemma~\ref{lem:2n-lr-prime-basics}
  $l_i(v)=l'_i(l_{i+1}(v))$, again contradicting
  our assumption.
  Since $\widehat{F}_{i+1}(v)\neq\emptyset$, we therefore have $(u,u')\in{}\widehat{R}_{i+1}(v)$.



  But then we can choose $P$ so it goes through $(m,m')$ where
  $m=r_{i+1}(v)\neq\bot$.  Now
  ${i+1}\leq{}d_2[w']\leq{}d_2[r_{i+1}(v)]\leq{}{i+1}$ so
  $d_2[m]={i+1}$.

  Let $e$ be the last edge in $\widehat{R}_i(v)$ then any path
  $\simplepath{r_i(v)}{v}$ that starts with $e$ crosses
  $P\cup\widehat{R}_{i+1}(v)$, implying that there exists such
  a path that contains $(m,m')$ and thus
  $r_i(v)=r_i(m)$.  Since $d_2[m]={i+1}$, then
  $l_i(v)=l'_i(m)$ and $r_i(v)=r'_i(m)$
  follows from lemma~\ref{lem:2n-lr-prime-basics}.
\end{proof}

\begin{definition}\label{def:2n-m}
  Let $v\in{}V$ and $0\leq{}i<d_2[v]$.
  \begin{align*}
    m_i(v)&:=
    \begin{cases}
      v&\text{ if $i+1=d_2[v]$}
      \\
      l_{i+1}(v)&\text{ if $i+1<d_2[v]\wedge{}r_i(v)\neq{}r'_i(r_{i+1}(v))$}
      \\
      r_{i+1}(v)&\text{ if $i+1<d_2[v]\wedge{}l_i(v)\neq{}l'_i(l_{i+1}(v))$}
      \\
      m_{i+1}(v)&\text{ otherwise}
    \end{cases}
  \end{align*}
\end{definition}

\begin{corollary}\label{cor:2n-m-crossing}
  Let $v\in{}V$ and $0\leq{}i<d_2[v]-1$.  If
  $l_i(v)\neq{}l'_i(l_{i+1}(v))$ or
  $r_i(v)\neq{}r'_i(r_{i+1}(v))$ then
  \begin{align*}
    l_i(v)&=l'_i(m_i(v))
    &&\wedge&
    r_i(v)&=r'_i(m_i(v))
    &&\wedge&
    d_2[m_i(v)]&=i+1
  \end{align*}
\end{corollary}
\begin{proof}
  This is just a reformulation of lemma~\ref{lem:2n-crossing} in terms
  of $m_i(v)$.
\end{proof}

\begin{lemma}\label{lem:2n-lrm-combined}
  For any vertex $v\in{}V$ and $0\leq{}i<d_2[v]$
  \begin{align*}
    l_i(v)&=l'_i(m_i(v))
    &&\wedge&
    r_i(v)&=r'_i(m_i(v))
  \end{align*}
\end{lemma}
\begin{proof}
  The proof is by induction on $j$, the number of times the
  ``otherwise'' case is used before reaching one of the other cases
  when expanding the recursive definition of $m_i(v)$.

  For $j=0$, either $i+1=d_2[v]$ and the result follows from
  Lemma~\ref{lem:2n-lr-prime-basics}, or $i+1<d_2[v]$ and
  $l_i(v)\neq{}l'_i(l_{i+1}(v))$ or $r_i(v)\neq{}r'_i(r_{i+1}(v))$.
  In either case we have by Corollary~\ref{cor:2n-m-crossing}, that
  $l_i(v)=l'_i(m_i(v))$ and
  $r_i(v)=r'_i(m_i(v))$.

  For $j>0$ we have $i+1<d_2[v]$ and $l_i(v)=l'_i(l_{i+1}(v))$ and
  $r_i(v)=r'_i(r_{i+1}(v))$ and $m_i(v)=m_{i+1}(v)$.  By induction we
  can assume that
  $l_{i+1}(v)=l'_{i+1}(m_{i+1}(v))$ and
  $r_{i+1}(v)=r'_{i+1}(m_{i+1}(v))$.
  Then by Lemma~\ref{lem:2n-lr-prime-reduce},
  $l'_i(l_{i+1}(v))=l'_i(l'_{i+1}(m_{i+1}(v)))=l'_i(m_{i+1}(v))=l'_i(m_i(v))$,
  showing that $l_i(v)=l'_i(m_i(v))$ as desired.
  The case for $r$ is symmetric.
\end{proof}


\begin{definition}\label{def:2n-pm}
  For any vertex $v\in{}V$, let
  \begin{align*}
    M[v]&:=\set{i\middle|\:0<i<d_2[v]\wedge{}m_{i-1}(v)\neq{}m_i(v)}
    \\
    p_m[v]&:=
    \begin{cases}
      \bot&\text{ if }M[v]=\emptyset
      \\
      m_{\max{}M[v]-1}(v)&\text{ otherwise}
    \end{cases}
  \end{align*}
  And define $T_m$ as the rooted forest over $V$ whose parent pointers
  are $p_m$.
\end{definition}

\begin{theorem}\label{thm:2n-frames}
  There exists a practical RAM data structure that for any good
  st-decomposition of a graph with $n$
  vertices uses $\Oo(n)$ words of $\Oo(\log{}n)$ bits and can answer
  $l_i(v)$ and $r_i(v)$ queries in constant time.
\end{theorem}
\begin{proof}
  For any vertex $v\in{}V$, let
  \begin{align*}
    D_l[v]&:=\set{i\middle|\:\text{$v$ has a proper ancestor $w$ in $T_l$
        with $d_2[w]=i$}}
    \\
    D_r[v]&:=\set{i\middle|\:\text{$v$ has a proper ancestor $w$ in $T_r$
        with $d_2[w]=i$}}
  \end{align*}
  Now, store levelancestor structures for each of $T_l$,
  $T_r$, and $T_m$, together with $d_2[v]$, 
  $D_l[v]$, $D_r[v]$, and $M[v]$ for each vertex.
  Since the height of the st-decomposition is $\Oo(\log{}n)$ each of
  $D_l[v]$, $D_r[v]$, and $M[v]$ can be
  represented in a single $\Oo(\log{}n)$-bit word.

  This representation allows us to find
  $d_2[m_i(v)]=\operatorname{succ}(M[v]\cup\set{d_2[v]},i)$ in
  constant time, as well as computing the depth in $T_m$ of $m_i(v)$.
  Then using the levelancestor structure for $T_m$ we can compute
  $m_i(v)$ in constant time.

  Similarly, this representation of the $D_l[v]$ set lets us compute
  the depth in $T_l$ of $l'_i(v)$ in constant time, and with the
  levelancestor structure that lets us compute $l'_i(v)$ in constant
  time.  A symmetric argument shows that we can compute $r'_i(v)$ in
  constant time.

  Finally, lemma~\ref{lem:2n-lrm-combined} says we can compute $l_i(v)$
  and $r_i(v)$ in constant time given constant-time functions for
  $l'$, $r'$, and $m$.
\end{proof}

%% file: 4-frame.tex
\subsection{4-frames}\label{sec:4-frames}

\begin{definition}
  Let $x$ be a node in an s-t-decomposition such that $F_x$ is a
  $4$-frame, and let $y$ be its parent.  Let $s^0_x$ and $s^1_x$ be
  the source corners on $F_x$ and let $t^0_x$ and $t^1_x$ be the
  target corners on $F_x$, numbered such that their clockwise cyclic
  order on $F_x$ is $s^0_x,t^0_x,s^1_x,t^1_x$, and such that if $F_y$
  is a $4$-frame there is an $\alpha\in\set{0,1}$ so $t^\alpha_x=t^\alpha_y$.
\end{definition}


\begin{definition}\label{def:4-global-LR-partition}
  Let $\mathcal{E}_4$ be the set of edges $(u,v)$ such if $x$ is the
  node in the s-t-decomposition that contains $v$, then $(u,v)\in{}E_x$ and $F_x$ is a
  $4$-frame. 
  Let $(\mathcal{L}^0,\mathcal{R}^0,\mathcal{L}^1,\mathcal{R}^1)$ be
  the partition of $\mathcal{E}_4$ defined as follows: For each
  $(u,v)\in\mathcal{E}_4$ let $x$ be the node such that $v\in{}C_x$,
  and let $y$ be the node (if it exists) closest to the root of
  $\mathcal{T}$ such that
  \begin{itemize}
  \item For any $z$ that is ancestor to $x$ and descendent to $y$, $F_z$ is a $4$-frame.
  \item $(u,v)\in{}E_y$.
  \item $u$ is not a target vertex of $F_y$.
  \end{itemize}
  If $y$ exists, then $(u,v)$ is incident to a corner $c$ on $F_y$.
  If there is an $\alpha\in\set{0,1}$ such that $c$ is on the
  clockwise disegment of $F_y$ between $s^\alpha_y$ and $t^\alpha_y$
  we assign $(u,v)$ to $\mathcal{R}^\alpha$.
  Otherwise there must be an
  $\alpha\in\set{0,1}$ such that $c$ is on the counterclockwise disegment of
  $F_y$ between $s^{1-\alpha}_y$ and $t^\alpha_y$, and we assign $(u,v)$ to
  $\mathcal{L}^\alpha$.
  If no such $y$ exists, $(u,v)$ must be incident to $t^\alpha_x$ for
  some $\alpha\in\set{0,1}$ and we (arbitrarily) assign $(u,v)$ to
  $\mathcal{L}^\alpha$.
\end{definition}

\begin{definition}
  Let $\mathcal{T}$ be an st-decomposition of $G=(V,E)$.  For any
  vertex $v\in{}V$ define:
  \begin{align*}
    c[v]&:=\text{The node }x\text{ in }\mathcal{T}\text{ such that }v\in{}V[C_x]
    \\
    d[v]&:=\text{The depth of }c[v]\text{ in }\mathcal{T}
    \\
    J_2[v]&:=\set{\operatorname{depth}(x)\middle|\:\text{$x$ is a non-root ancestor to $c[v]$ in $\mathcal{T}$ and $F_x$ is a $2$-frame}}
    \\
    j_2[v]&:=\operatorname{max}(J_2[v]))
  \end{align*}
 \end{definition} 
  
The number $j_2[v]$ is especially useful for $4$-frame nodes. On the path from the root to the component of $v$ in the s-t-decomposition tree, there will be a last component whose frame is a $2$-frame. We call the depth of the next component on the path $j_2[v]$. If $c[v]$ has a $4$-frame, then for the rest of the path, that is, depth $i$ with $j_2[v]\leq i <d[v]$, we will have $4$-frames nested in $4$-frames, which gives a lot of useful structure.

\begin{definition}\label{def:4-i-of-v}
  For any $j_2[v]\leq{}i<d[v]$ and $\alpha\in\set{0,1}$, let $x$ be the
  ancestor of $c[v]$ at depth $i+1$ and define:
  \begin{align*}
    E_i(v)&:=E_x
    \\
    L^\alpha_i(v)&:=E_x\cap\mathcal{L}^\alpha
    \\
    R^\alpha_i(v)&:=E_x\cap\mathcal{R}^\alpha
    \\
    \widehat{L}^\alpha_i(v)&:=\set{(w,w')\in{}L^\alpha_i(v)\middle|\:
      \reachable{w'}{v}
    }
    \\
    \widehat{R}^\alpha_i(v)&:=\set{(w,w')\in{}R^\alpha_i(v)\middle|\:
      \reachable{w'}{v}
    }
    \\
    \widehat{F}_i(v)&:=\widehat{L}^0_i(v)\cup\widehat{R}^0_i(v)\cup\widehat{L}^1_i(v)\cup\widehat{R}^1_i(v)
    \\
    l^\alpha_i(v)&:=
    \begin{cases}
      \bot&\text{ if }\widehat{L}^\alpha_i(v)=\emptyset
      \\
      \text{the last vertex in $\source(\widehat{L}^\alpha_i(v))$ on the counterclockwise dipath of $F_x$}&\text{ otherwise}      
    \end{cases}
    \\
    r^\alpha_i(v)&:=
    \begin{cases}
      \bot&\text{ if }\widehat{R}^\alpha_i(v)=\emptyset
      \\
      \text{the last vertex in $\source(\widehat{R}^\alpha_i(v))$ on the clockwise dipath of $F_x$}&\text{ otherwise}      
    \end{cases}
    \\
    s^\alpha_i(v)&:=\text{The vertex associated with }s^\alpha_x
    \\
    t^\alpha_i(v)&:=\text{The vertex associated with }t^\alpha_x
  \end{align*}
  Additionally, let $L^\alpha_i(v)$ and $\widehat{L}^\alpha_i(v)$ be totally ordered
  by the position of the starting vertices on the counterclockwise
  disegment of $F_x$ and the clockwise order around each starting
  vertex.  Similarly let $R^\alpha_i(v)$ and $\widehat{R}^\alpha_i(v)$ be totally
  ordered by the position of the starting vertices on the
  clockwise disegment of $F_x$ and the counterclockwise order around
  each starting vertex.
\end{definition}

We know from Section~\ref{sec:2-frames} that we can find the relevant
vertices on each $2$-frame surrounding $v$.  The goal in this section
is a data structure for efficiently computing
$l^\alpha_i(v)$ and
$r^\alpha_i(v)$ for $j_2[v]\leq{}i<d[v]$.

\begin{lemma}\label{lem:4-hat-nonempty}
  For any vertex $v\in{}V$ and $j_2[v]\leq{}i<d[v]$:
  $\widehat{F}_i(v)\neq\emptyset$
\end{lemma}
\begin{proof}
  Let $x$ be the ancestor of $c[v]$ at depth $i+1$.  Since $G$ is a
  single-source graph, there is a path from $s$ to $v$.  This path
  must contain a vertex in $V[F_x]$. 
  But then the edge following the last such
  vertex on the path must be in
  $\widehat{L}^0_i(v)\cup{}\widehat{R}^0_i(v)\cup{}\widehat{L}^1_i(v)\cup{}\widehat{R}^1_i(v)$
  which is therefore nonempty.
\end{proof}

\begin{lemma}\label{lem:4-hat-subset}
  For any $u,v\in{}V$, $j_2[v]\leq{}i<d[u]$, and $\alpha\in\set{0,1}$: If $\reachable{u}{v}$ then
  $\widehat{L}^\alpha_i(u)\subseteq\widehat{L}^\alpha_i(v)$ and $\widehat{R}^\alpha_i(u)\subseteq\widehat{R}^alpha_i(v)$.
\end{lemma}
\begin{proof}
  Since $\reachable{u}{v}$, $c[u]$ is ancestor to $c[v]$ and so
  $L^\alpha_i(u)=L^\alpha_i(v)$ and hence
  $\widehat{L}^\alpha_i(u)\subseteq\widehat{L}^\alpha_i(v)$.  Similarly,
  $R^\alpha_i(u)=R^\alpha_i(v)$ and $\widehat{R}^\alpha_i(u)\subseteq\widehat{R}^\alpha_i(v)$.
\end{proof}



\begin{lemma}\label{lem:4-lr-multiframe}
  Given any vertex $v\in{}V$, $j_2[v]\leq{}i<d[v]$,
  $\alpha\in\set{0,1}$, and $(w,w')\in{}E_i(v)$.  Then:
  \begin{align*}
    (w,w')&\in\widehat{L}^\alpha_i(v)
    &&\implies&
    (w,w')&\in\widehat{L}^\alpha_{i'}(v)\text{ for all }i^\prime, \max\set{d[w],j_2[v]}\le i^\prime <\min\set{d[w'],d[v]}
    \\
    (w,w')&\in\widehat{R}^\alpha_i(v)
    &&\implies&
    (w,w')&\in\widehat{R}^\alpha_{i'}(v)\text{ for all }i^\prime, \max\set{d[w],j_2[v]}\le i^\prime <\min\set{d[w'],d[v]}
  \end{align*}
\end{lemma}
\begin{proof}
  Let $j=\max\set{d[w],j_2[v]}$ and $k=\min\set{d[w'],d[v]}$.
  Clearly $(w,w')\in{}E_{i'}$ for all $j\leq{}i'<k$.  Suppose
  $(w,w')\in\widehat{L}^\alpha_i(v)\subseteq{}L^\alpha_i(v)$, then
  since $j\leq{}i<k$ the definition give us
  $(w,w')\in{}L^\alpha_{i'}(v)$ for all $j\leq{}i'<k$.  And since
  $\reachable{w'}{v}$ this implies
  $(w,w')\in\widehat{L}^\alpha_{i'}(v)$ for all $j\leq{}i'<k$ and the result follows.
  The case for $R$ is symmetric.
\end{proof}


\begin{definition}\label{def:4-pl-pr}
  For any vertex $v\in{}V$ and $\alpha\in\set{0,1}$ let
  \begin{align*}
    p^\alpha_l[v]&:=
    \begin{cases}
      \bot&\text{ if }d[v]=0\vee\text{ $F_{d[v]-1}(v)$ is a $2$-frame}
      \\
      l^\alpha_{d[v]-1}(v)&\text{ otherwise}
    \end{cases}
    \\
    p^\alpha_r[v]&:=
    \begin{cases}
      \bot&\text{ if }d[v]=0\vee\text{ $F_{d[v]-1}(v)$ is a $2$-frame}
      \\
      r^\alpha_{d[v]-1}(v)&\text{ otherwise}
    \end{cases}
  \end{align*}
  and let $T^\alpha_l$ and $T^\alpha_r$ denote the rooted forests over $V$ whose
  parent pointers are $p^\alpha_l$ and $p^\alpha_r$ respectively.
\end{definition}

\begin{definition}\label{def:4-lr-prime}
  For any $v\in{}V\cup\set{\bot}$, $\alpha\in\set{0,1}$, and $i\geq{}j_2[v]$ let
  \begin{align*}
    l'^\alpha_i(v)&:=
    \begin{cases}
      v&\text{ if $v=\bot\vee{}d[v]\leq{}i$}
      \\
      l'^\alpha_i(p^\alpha_l[v])&\text{ otherwise}
    \end{cases}
    \\
    r'^\alpha_i(v)&:=
    \begin{cases}
      v&\text{ if $v=\bot\vee{}d[v]\leq{}i$}
      \\
      r'^\alpha_i(p^\alpha_r[v])&\text{ otherwise}
    \end{cases}
  \end{align*}
\end{definition}

\begin{lemma}\label{lem:4-lr-prime-basics}
  Let $v\in{}V$, $\alpha\in\set{0,1}$, and $i\geq{}j_2[v]$ be given, then
  \begin{align*}
    i&=d[v]-1
    &&\implies&
    l'^\alpha_i(v)&=l^\alpha_i(v)
    &&\wedge&
    r'^\alpha_i(v)&=r^\alpha_i(v)
    \\
    i&\leq{}d[v]-1
    &&\implies&
    l'^\alpha_i(v)&\in\source(\widehat{L}^\alpha_i(v))\cup\set{\bot}
    &&\wedge&
    r'^\alpha_i(v)&\in\source(\widehat{R}^\alpha_i(v))\cup\set{\bot}
    \\
    i&>d[v]-1
    &&\implies&
    l'^\alpha_i(v)&=v
    &&\wedge&
    r'^\alpha_i(v)&=v
  \end{align*}
\end{lemma}
\begin{proof}
  We will show this for $l'$ only, as $r'$ is completely symmetrical.
  If $i>d[v]-1$ then $d[v]\leq{}i$ and we get $l'^\alpha_i(v)=v$ directly
  from the definition of $l'$.
  Similarly if $i=d[v]-1$ then $l'^\alpha_i(v) =
  l'^\alpha_i(p^\alpha_l[v]) =
  l'^\alpha_i(l^\alpha_{d[v]-1}(v)) =
  l'^\alpha_i(l^\alpha_i(v)) =
  l^\alpha_i(v) \in
  \source(\widehat{L}^\alpha_i(v))\cup\set{\bot}$.
  Finally suppose $i<d[v]-1$.  If $l'^\alpha_i(v)=\bot$ we are done,
  so suppose that is not the case.  Let $u$ be the child of
  $l'^\alpha_i(v)$ in $T_l$ that is ancestor to $v$.  Then
  $l'^\alpha_i(v) = l'^\alpha_i(u) = p^\alpha_l[u] =
  l^\alpha_{d[u]-1}(u)$.  By definition of $l^\alpha_{d[u]-1}(u)$ there
  exists an edge $(w,w')\in\widehat{L}^\alpha_{d[u]-1}$ where
  $w=l^\alpha_{d[u]-1}(u)$ and $d[w]\leq{}i<d[w']\leq{}d[u]$ and by
  setting $(v,i,(w,w')) = (u,d[u]-1,(w,w'))$ in
  lemma~\ref{lem:4-lr-multiframe} we get
  $(w,w')\in\widehat{L}^\alpha_i(u)$, and therefore
  $l'^\alpha_i(v)\in\source(\widehat{L}^\alpha_i(u))$.  But since
  $\reachable{u}{v}$ we have
  $\widehat{L}^\alpha_i(u)\subseteq\widehat{L}^\alpha_i(v)$ by Lemma~\ref{lem:4-hat-subset} and we are
  done.
\end{proof}

\begin{lemma}\label{lem:4-lr-prime-reduce}
  Let $v\in{}V$, $\alpha\in\set{0,1}$, and $j_2[v]\leq{}i\leq{}j$ then
  \begin{align*}
    l'^\alpha_i(l'^\alpha_j(v))&=l'^\alpha_i(v)
    &&\wedge&
    r'^\alpha_i(r'^\alpha_j(v))&=r'^\alpha_i(v)
  \end{align*}
\end{lemma}
\begin{proof}
  $l'^\alpha_j(v)$ is on the path from $v$ to $l'^\alpha_i(v)$ in $T_l$, so this
  follows trivially from the recursion.  The case for $r'$ is symmetric.
\end{proof}

\begin{lemma}\label{lem:4-lr-bot}
  Let $v\in{}V$, $\alpha\in\set{0,1}$, and $j_2[v]\leq{}i<d[v]-1$, then
  \begin{alignat*}{7}
    l^\alpha_i(v)&=\bot
    &&\quad\implies\quad&
    l'^\alpha_i(l^\alpha_{i+1}(v))&=\bot
    &&\qquad\wedge\qquad&
    r^\alpha_i(v)&=\bot
    &&\quad\implies\quad&
    r'^\alpha_i(r^\alpha_{i+1}(v))&=\bot
  \end{alignat*}
\end{lemma}
\begin{proof}
  If $l^\alpha_i(v)=\bot$ then $\widehat{L}^\alpha_i(v)=\emptyset$, so
  either $l^\alpha_{i+1}(v)=\bot$ implying
  $l'^\alpha_i(l^\alpha_{i+1}(v))=\bot$ by the definition of $l'$, or
  $l^\alpha_{i+1}(v)\not\in\source(\widehat{L}^\alpha_i(v))$ so
  $d[l^\alpha_{i+1}(v)]=i+1$ and by Lemma~\ref{lem:4-lr-prime-basics}
  $l'^\alpha_i(l^\alpha_{i+1}(v)) \in
  \source(\widehat{L}^\alpha_i(l^\alpha_{i+1}(v)))\cup\set{\bot}
  \subseteq \source(\widehat{L}^\alpha_i(v))\cup\set{\bot}=\set{\bot}$
  so again $l'^\alpha_i(l^\alpha_{i+1}(v))=\bot$.  The case for $r$ is
  symmetric.
\end{proof}


\begin{lemma}[Crossing lemma]\label{lem:4-crossing}
  Let $v\in{}V$, $\alpha\in\set{0,1}$, and $j_2[v]\leq{}i<d[v]-1$.
  \begin{alignat*}{7}
    l^\alpha_i(v)&\neq{}l'^\alpha_i(l^\alpha_{i+1}(v))
    &&\implies\quad&
    l^\alpha_i(v)&=l'^\alpha_i(m)
    &&\:\wedge\:&
    r^\alpha_i(v)&=r'^\alpha_i(m)
    &&\:\wedge\:&
    d[m]&=i+1
    \\&&&\text{where }m=r^\alpha_{i+1}(v)\neq\bot\hspace{-30cm}
    \\
    r^\alpha_i(v)&\neq{}r'^\alpha_i(r^\alpha_{i+1}(v))
    &&\implies\quad&
    l^\alpha_i(v)&=l'^\alpha_i(m)
    &&\:\wedge\:&
    r^\alpha_i(v)&=r'^\alpha_i(m)
    &&\:\wedge\:&
    d[m]&=i+1
    \\&&&\text{where }m=l^\alpha_{i+1}(v)\neq\bot\hspace{-30cm}
  \end{alignat*}
\end{lemma}
\begin{figure}
\centering
\includegraphics[width=0.7\linewidth]{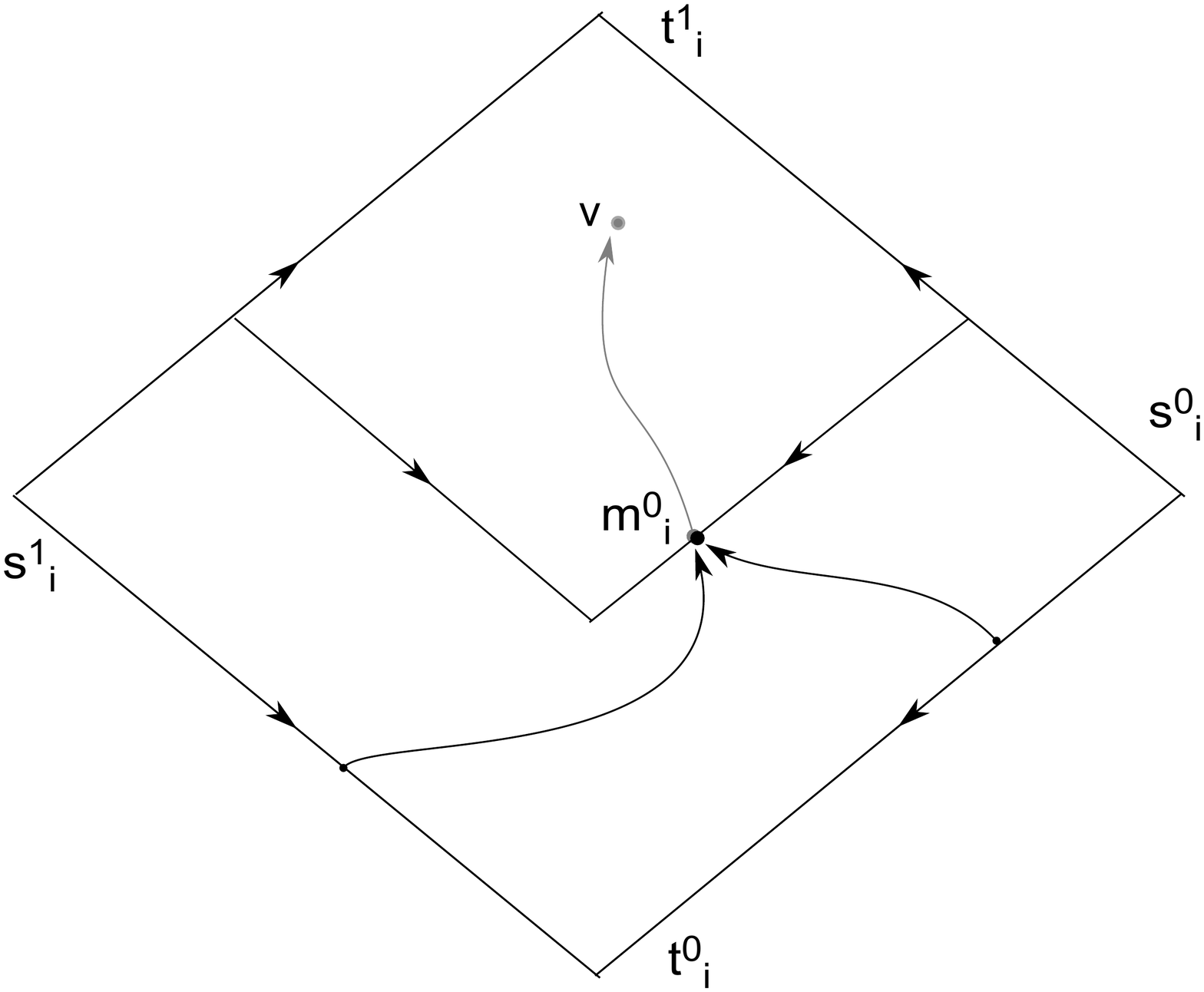}
\caption{Sometimes the best path from $L^0_i(v)$ to $v$ must go
  through $R^0_{i+1}(v)$.}
\label{fig:4frame-crossing}
\end{figure}
\begin{proof}
  Suppose $l^\alpha_i(v)\neq{}l'^\alpha_i(l^\alpha_{i+1}(v))$ (the case
  $r^\alpha_i(v)\neq{}r'^\alpha_i(r^\alpha_{i+1}(v))$ is symmetrical).
  Then $l^\alpha_i(v)\neq\bot$ by lemma~\ref{lem:4-lr-bot}.
  Thus there is a last edge $(w,w')\in\widehat{L}^\alpha_i(v)$ with
  $w=l^\alpha_i(v)$ and $d[w]\leq{}i<d[w']$ and a path
  $P=\simplepath{w'}{v}$.

  Now $(w,w')\not\in{}E_{i+1}(v)$
  since otherwise by Definition~\ref{def:4-i-of-v}
  $(w,w')\in{}L^\alpha_{i+1}(v)$ and since $\reachable{w'}{v}$ even
  $(w,w')\in\widehat{L}^\alpha_{i+1}(v)$ implying
  $l^\alpha_i(v)=l^\alpha_{i+1}(v)$ and thus
  $l^\alpha_i(v)=l'^\alpha_i(l^\alpha_{i+1}(v))$ by
  lemma~\ref{lem:4-lr-prime-basics}, contradicting our assumption.

  Since $(w,w')\not\in{}E_{i+1}(v)$, the path $P$ must cross
  $\widehat{F}_{i+1}(v)$.  Let $(u,u')$ be the last edge in
  $P\cap{}\widehat{F}_{i+1}(v)$.  Then $\reachable{w'}{u}$ so
  $d[u]\geq{}i+1$ and $(u,u')\not\in{}L^\alpha_{i+1}(v)$ since
  otherwise $d[l^\alpha_{i+1}(v)]=i+1$ and hence by
  Lemma~\ref{lem:4-lr-prime-basics}
  $l^\alpha_i(v)=l'^\alpha_i(l^\alpha_{i+1}(v))$, again contradicting
  our assumption.

  Also, $t^\alpha_i(v)\neq{}t^\alpha_{i+1}(v)$ because
  $t^\alpha_i(v)=t^\alpha_{i+1}(v)$ would imply 
  $(w,w')\in{}L^\alpha_{i+1}(v)\cup\set\bot$
  which we have just shown is not the case.

  Since $t^\alpha_i(v)\neq{}t^\alpha_{i+1}(v)$, then by definition
  $t^{1-\alpha}_i(v)=t^{1-\alpha}_{i+1}(v)$ and hence 
  $L^{1-\alpha}_{i+1}(v)\subseteq{}L^{1-\alpha}_{i}(v)$ and
  $R^{1-\alpha}_{i+1}(v)\subseteq{}R^{1-\alpha}_{i}(v)$, implying
  $d[w'']\leq{}i$ for all
  $w''\in{}L^{1-\alpha}_{i+1}(v)\cup{}R^{1-\alpha}_{i+1}(v)$.
  Thus,
  $(u,u')\not\in{}L^{1-\alpha}_{i+1}(v)\cup{}R^{1-\alpha}_{i+1}(v)$ since
  $d[u]>i$, and we can conclude that $(u,u')\in{}\widehat{R}^\alpha_{i+1}(v)$.

  But then we can choose $P$ so it goes through $(m,m')$ where
  $m=r^\alpha_{i+1}(v)\neq\bot$.  Now
  ${i+1}\leq{}d[w']\leq{}d[r^\alpha_{i+1}(v)]\leq{}{i+1}$ so
  $d[m]={i+1}$.

  Let $e$ be the last edge in $\widehat{R}^\alpha_i(v)$ then any path
  $\simplepath{r^\alpha_i(v)}{v}$ that starts with $e$ crosses
  $P\cup\widehat{R}^\alpha_{i+1}(v)$, implying that there exists such
  a path that contains $(m,m')$ and thus
  $r^\alpha_i(v)=r^\alpha_i(m)$.  Since $d[m]={i+1}$, then
  $l^\alpha_i(v)=l'^\alpha_i(m)$ and $r^\alpha_i(v)=r'^\alpha_i(m)$
  follows from lemma~\ref{lem:4-lr-prime-basics}.
\end{proof}

\begin{definition}\label{def:4-m}
  Let $v\in{}V$, $\alpha\in\set{0,1}$, and $0\leq{}i<d[v]$.
  \begin{align*}
    m^\alpha_i(v)&:=
    \begin{cases}
      v&\text{ if $i+1=d[v]$}
      \\
      l^\alpha_{i+1}(v)&\text{ if $i+1<d[v]\wedge{}r^\alpha_i(v)\neq{}r'^\alpha_i(r^\alpha_{i+1}(v))$}
      \\
      r^\alpha_{i+1}(v)&\text{ if $i+1<d[v]\wedge{}l^\alpha_i(v)\neq{}l'^\alpha_i(l^\alpha_{i+1}(v))$}
      \\
      m^\alpha_{i+1}(v)&\text{ otherwise}
    \end{cases}
  \end{align*}
\end{definition}

\begin{corollary}\label{cor:4-m-crossing}
  Let $v\in{}V$, $\alpha\in\set{0,1}$, and $j_2[v]\leq{}i<d[v]-1$.  If
  $l^\alpha_i(v)\neq{}l'^\alpha_i(l^\alpha_{i+1}(v))$ or
  $r^\alpha_i(v)\neq{}r'^\alpha_i(r^\alpha_{i+1}(v))$ then
  \begin{align*}
    l^\alpha_i(v)&=l'^\alpha_i(m^\alpha_i(v))
    &&\wedge&
    r^\alpha_i(v)&=r'^\alpha_i(m^\alpha_i(v))
    &&\wedge&
    d[m^\alpha_i(v)]&=i+1
  \end{align*}
\end{corollary}
\begin{proof}
  This is just a reformulation of lemma~\ref{lem:4-crossing} in terms
  of $m^\alpha_i(v)$.
\end{proof}

\begin{lemma}\label{lem:4-lrm-combined}
  For any vertex $v\in{}V$, $\alpha\in\set{0,1}$, and $j_2[v]\leq{}i<d[v]$
  \begin{align*}
    l^\alpha_i(v)&=l'^\alpha_i(m^\alpha_i(v))
    &&\wedge&
    r^\alpha_i(v)&=r'^\alpha_i(m^\alpha_i(v))
  \end{align*}
\end{lemma}
\begin{proof}
  The proof is by induction on $j$, the number of times the
  ``otherwise'' case is used before reaching one of the other cases
  when expanding the recursive definition of $m_i(v)$.

  For $j=0$, either $i+1=d[v]$ and the result follows from
  Lemma~\ref{lem:4-lr-prime-basics}, or $i+1<d[v]$ and
  $l_i(v)\neq{}l'_i(l_{i+1}(v))$ or $r_i(v)\neq{}r'_i(r_{i+1}(v))$.
  In either case we have by Corollary~\ref{cor:4-m-crossing}, that
  $l^\alpha_i(v)=l'^\alpha_i(m^\alpha_i(v))$ and
  $r^\alpha_i(v)=r'^\alpha_i(m^\alpha_i(v))$.

  For $j>0$ we have $i+1<d[v]$ and $l_i(v)=l'_i(l_{i+1}(v))$ and
  $r_i(v)=r'_i(r_{i+1}(v))$ and $m_i(v)=m_{i+1}(v)$.  By induction we
  can assume that
  $l^\alpha_{i+1}(v)=l'^\alpha_{i+1}(m^\alpha_{i+1}(v))$ and
  $r^\alpha_{i+1}(v)=r'^\alpha_{i+1}(m^\alpha_{i+1}(v))$.
  Then by Lemma~\ref{lem:4-lr-prime-reduce},
  $l'^\alpha_i(l^\alpha_{i+1}(v))=l'^\alpha_i(l'^\alpha_{i+1}(m^\alpha_{i+1}(v)))=l'^\alpha_i(m^\alpha_{i+1}(v))=l'^\alpha_i(m^\alpha_i(v))$,
  showing that $l^\alpha_i(v)=l'^\alpha_i(m^\alpha_i(v))$ as desired.
  The case for $r$ is symmetric.
\end{proof}


\begin{definition}\label{def:4-pm}
  For any vertex $v\in{}V$, and $\alpha\in\set{0,1}$ let
  \begin{align*}
    M^\alpha[v]&:=\set{i\middle|\:j_2[v]<i<d[v]\wedge{}m^\alpha_{i-1}(v)\neq{}m^\alpha_i(v)}
    \\
    p^\alpha_m[v]&:=
    \begin{cases}
      \bot&\text{ if }M^\alpha[v]=\emptyset
      \\
      m^\alpha_{\max{}M^\alpha[v]-1}(v)&\text{ otherwise}
    \end{cases}
  \end{align*}
  And define $T^\alpha_m$ as the rooted forest over $V$ whose parent pointers
  are $p^\alpha_m$.
\end{definition}

\begin{theorem}\label{thm:4-frames}
  There exists a practical RAM data structure that for any good
  st-decomposition of a graph with $n$
  vertices uses $\Oo(n)$ words of $\Oo(\log{}n)$ bits and can answer
  $l^\alpha_i(v)$ and $r^\alpha_i(v)$ queries in constant time.
\end{theorem}
\begin{proof}
  For any vertex $v\in{}V$, and $\alpha\in\set{0,1}$ let
  \begin{align*}
    D^\alpha_l[v]&:=\set{i\middle|\:\text{$v$ has a proper ancestor $w$ in $T^\alpha_l$
        with $d[w]=i$}}
    \\
    D^\alpha_r[v]&:=\set{i\middle|\:\text{$v$ has a proper ancestor $w$ in $T^\alpha_r$
        with $d[w]=i$}}
  \end{align*}
  Now, store levelancestor structures for each of $T^\alpha_l$,
  $T^\alpha_r$, and $T^\alpha_m$, together with $d[v]$, $j_2[v]$,
  $J_2[v]$, $D^\alpha_l[v]$, $D^\alpha_r[v]$, and $M^\alpha[v]$ for
  each vertex.  Since the height of the st-decomposition is
  $\Oo(\log{}n)$ each of $J_2[v]$, $D^\alpha_l[v]$, $D^\alpha_r[v]$, and
  $M^\alpha[v]$ can be represented in a single $\Oo(\log{}n)$-bit
  word.

  This representation allows us to find
  $d[m^\alpha_i(v)]=\operatorname{succ}(M^\alpha[v]\cup\set{d[v]},i)$ in
  constant time, as well as computing the depth in $T^\alpha_m$ of $m^\alpha_i(v)$.
  Then using the levelancestor structure for $T^\alpha_m$ we can compute
  $m^\alpha_i(v)$ in constant time.

  Similarly, this representation of the $D^\alpha_l[v]$ set lets us compute
  the depth in $T^\alpha_l$ of $l'^\alpha_i(v)$ in constant time, and with the
  levelancestor structure that lets us compute $l'^\alpha_i(v)$ in constant
  time.  A symmetric argument shows that we can compute $r'^\alpha_i(v)$ in
  constant time.

  Finally, lemma~\ref{lem:4-lrm-combined} says we can compute $l^\alpha_i(v)$
  and $r^\alpha_i(v)$ in constant time given constant-time functions for
  $l'$, $r'$, and $m$.
\end{proof}

%% file: InOut.tex
\section{Acyclic planar In- out- graphs}\label{sec:inout}

For an in-out-graph $G$ we have a source, $s$, that can reach 
all vertices of outdegree $0$. Given such a source, $s$, we may assign all vertices a colour: A vertex is green if it can be reached from $s$, and red otherwise. We may also colour the directed edges: $(u,v)$ has the same colour as its endpoints, or is a blue edge in the special case where $u$ is red and $v$ is green. Our idea is to keep the colouring and flip all non-green edges, thus obtaining a single source graph $H$ with source $s$. (Any vertex was either green and thus already reachable from $s$, or could reach some target $t$, and is reachable from $s$ in $H$ via the first green vertex on its path to $t$.)

Consider the single source reachability data structure for the red-green graph, $H$. This alone does not suffice to determine reachability in $G$, but it does when endowed with a few extra words per vertex:
\begin{description}
\item[M1] A red vertex $\redu$ must remember the additional information of the best green vertices $BestGreen(\redu)$ on its own parent frame it can reach. There are at most $4$ such vertices, one for each disegment. \label{M1}
\item[M2] Information about paths from a red to a green vertex in the same component. See Section \ref{intra}. 
\item[M3] Information about paths from a red vertex in some component $C$ to a green vertex in an ancestor component of $C$. See Section \ref{inter}.
\end{description}

Given a green vertex $\greencol{v}$, we know for each ancestral frame segment the best vertex that can reach $\greencol{v}$. For a red vertex $\redcol{u}$,
given a segment $p$ on an ancestral frame to 
$\redcol{u}$, 
we have information about the best vertex on $p$ that may reach $u$ in $H$ via ``ingoing'' edges, that is, an edge from the corresponding $\widehat{F}_i(u)$. 
If that best vertex is red, then it is the best vertex on $p$ that $\redcol{u}$ can reach, again, from the ``inside''.

We may now case reachability based on the colour of nodes:
\begin{itemize}
\item For green $\greencol{u}$ and red $\redcol{v}$,  reach$_G$($\greencol{u},\redcol{v}$) = No.
\item For green vertices $\greencol{u},\greencol{v}$, reach$_G(\greencol{u},\greencol{v})$ = reach$_H(u,v)$
\item For red vertices $\redcol{u},\redcol{v}$, reach$_G(\redcol{u},\redcol{v})$ = reach$_H(v,u)$
\item When $\redu$ is red and $\greenv$ is green, to determine reach$_G(\redcol{u},\greencol{v})$ we need more work. It will depend on where in the hierarchy of components, $u$ and $v$ reside. 
\end{itemize}
%
\begin{wrapfigure}[16]{r}{0.24\textwidth}
\includegraphics[width=1.0\linewidth]{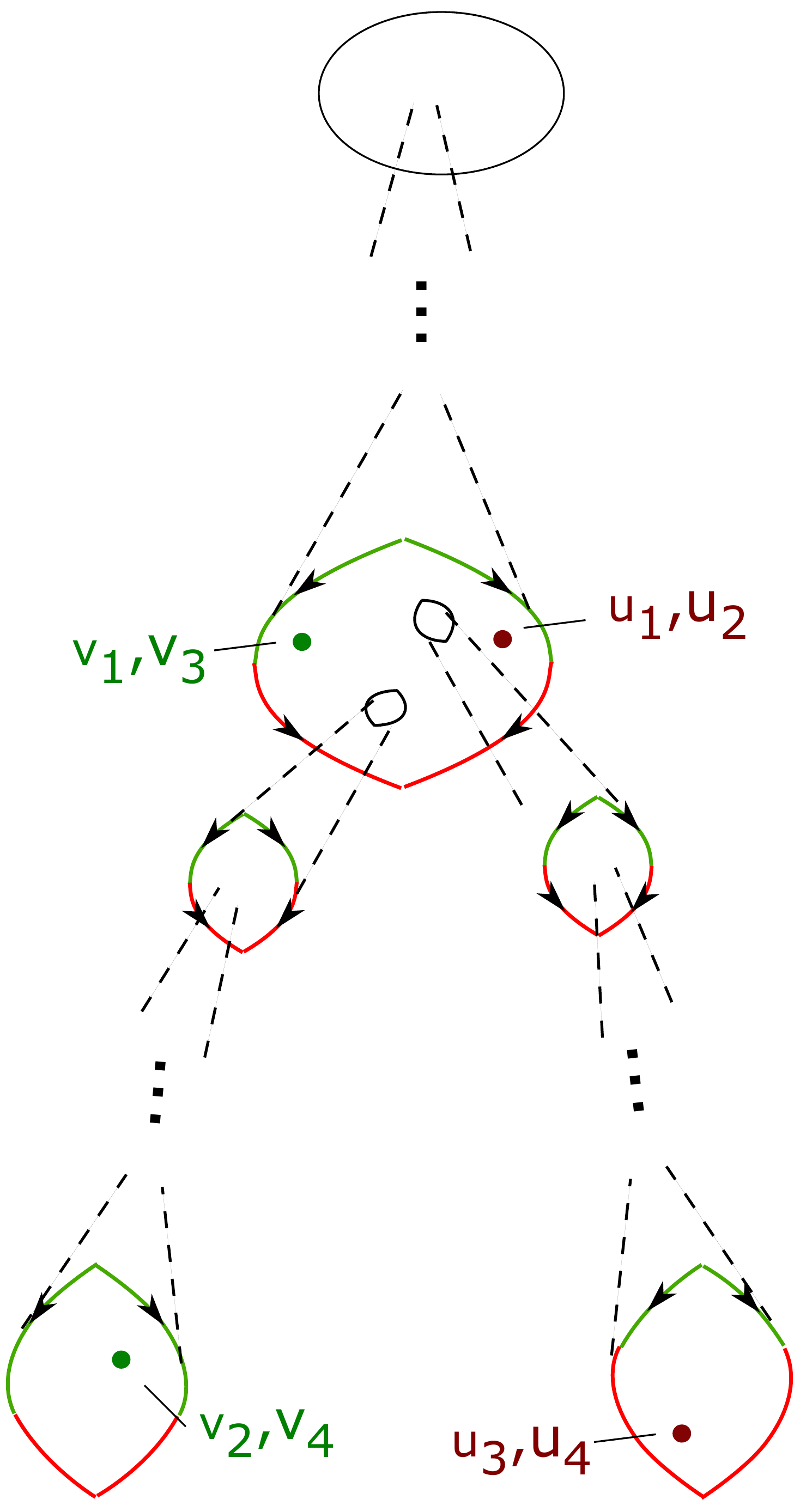}
\label{fig:red-green}
\end{wrapfigure}

When $\redu$ is red and $\greenv$ is green, there are the following cases.

\begin{enumerate}
\item $c[\redcol{u}]=c[\greencol{v}]$. There may be a path from $\redu$ to $\greenv$:
\begin{itemize}[label=$\circ$,leftmargin=*,noitemsep,topsep=0pt,parsep=0pt,partopsep=0pt]
\item Via a green vertex $\greencol{w}$ in the parent frame of $\redcol{u}$. For each candidate $\greenw\in BestGreen(\redu)$, try reach$_H(\greencol{w},\greencol{v})$. (See M1).
\item Staying within the frame, that is, reach$_{c[\redcol{u}]}(\redcol{u},\greencol{v})$. To handle this case we need to store more information, see Section \ref{intra}.
\end{itemize}
\item $c[\redcol{u}] \prec c[\greencol{v}]$. There may be a path from $\redu$ to $\greenv$:
\begin{itemize}[label=$\circ$,leftmargin=*,noitemsep,topsep=0pt,parsep=0pt,partopsep=0pt]
\item Via a green vertex $\greencol{w}$ in the parent frame of $\redcol{u}$, reach$_H(\greencol{w},\greencol{v})$. (See M1).
\item Via a green vertex $\greencol{w}$, where $c[\greencol{w}]=c[\redcol{u}]$, then reach$_G (\redcol{u},\greencol{w})$ is in case 1 above.
$\greencol{v}$ knows the at most $4$ such $\greencol{w}$s from the single source structure. 
\end{itemize}
\item $c[\redcol{u}]\succ c[\greencol{v}]$. There may be a path from $\redu$ to $\greenv$:
\begin{itemize}[label=$\circ$,leftmargin=*,noitemsep,topsep=0pt,parsep=0pt,partopsep=0pt]
\item  Via a red edge $(\redcol{w^{\prime}}, \redcol{w})$ in $G$ with $c[\redcol{w}]\preceq c[\greencol{v}] \prec c[\redcol{w^{\prime}}] \preceq c[\redcol{u}]$. That is, in the single-source structure for $H$, $u$ can find its best vertex $w$ for each disegment of the parent frame of $c[\greenv]$. For a path via that disegment to exist, $w$ must be red, and reach$_G(\redcol{w},\greencol{v})$, which is in case 1 or 2 above, must return true. 
\item  Via a blue edge $(\redcol{w^\prime}, \greencol{w})$ with $c[\greencol{w}]\preceq c[\greencol{v}] \prec c[\redcol{w^{\prime}}] \preceq c[\redcol{u}]$. 
We handle this case in Section~\ref{inter}.
\end{itemize}
\newcounter{enumTemp}
\setcounter{enumTemp}{\theenumi}
\end{enumerate}

\begin{enumerate}
\setcounter{enumi}{\theenumTemp}
\item $c[\redcol{u}],c[\greencol{v}]\succ N$, where $N = \operatorname{lca}(c[\redcol{u}],c[\greencol{v}])$. A path from $\redu$ to $\greenv$ must go:
\begin{itemize}[label=$\circ$,leftmargin=*,noitemsep,topsep=0pt,parsep=0pt,partopsep=0pt]
\item Via $\greencol{w}$, $c[\greencol{w}]\preceq N$, then reach$_G(\redcol{u},\greencol{w})$ is in case 3 above.
$\greencol{v}$ computes at most $4$ such $\greencol{w}$s from the single source structure, and note that all the vertices that $\greenv$ computes must be green.
\end{itemize}

\end{enumerate}

\subsection{Intracomponental blue edges} \label{intra}

Consider the set of ``blue'' edges $(\redcol{a},\greencol{b})$ from $G$ where both the red vertex $\redcol{a}$ and green $\greencol{b}$ reside in some given component in the s-t-decomposition of $H$.

\begin{lemma}\label{lem:intrablue}
We may assign to each vertex $\le 2$ numbers,
such that
if red $\redcol{u}$ remembers $i,j\in \mathbb{N}$ and green $\greencol{v}$ remembers $l,r\in \mathbb{N}$, then $\redcol{u}$ can reach $\greencol{v}$ if and only if
$i\leq l\leq j$ or $i\leq r\leq j$ or $\min\!\set{l,r}\leq j<i$ or $j<i\leq\max\!\set{l,r}$.
\end{lemma}

\begin{proof}
The key observation is that we may enumerate all blue edges $b_0 = (\redcol{u_0},\greencol{v_0}), \ldots b_i = (\redcol{u_m},\greencol{v_m})$ such that any red vertex can reach a segment of their endpoints, $\greencol{v_i},\ldots,\greencol{v_j}$. 
Namely, the blue edges form a minimal cut in the planar graph which separates the red from the green vertices, and this cut induces a cyclic order. In this order, each red vertex may reach a segment of blue edges, and each green vertex may reach a segment of blue edge endpoints. Thus, the blue edge endpoints reachable from a given red vertex (through any path) is a union of overlapping segments, which is again a segment.

Now each red vertex remembers the indices of the first $\greencol{v_i}$ and last $\greencol{v_j}$ blue edge endpoint it may reach. 
For a green vertex $\greencol{v}$, the s-t-subgraph with $\greencol{v}$ as target has a delimiting face consisting of two paths, $P$ and $Q$. $\greencol{v}$ remembers the indices $l$, $r$ of the latest blue edge endpoints $v_l\in P$ and $v_r\in Q$, if they exist.
Clearly, if $l$ or $r$ is within range, $\redcol{u}$ may reach $\greencol{v}$. Contrarily, if $\redcol{u}$ may reach $\greencol{v}$, it must do so via some vertex 
$\greencol{v^\prime}$ on $P \cup Q$. But then $\greencol{v^\prime}$ must be able to reach $v_l$ or $v_r$, and thus, $l$ or $r$ is within range.
\end{proof}

\subsection{Intercomponental blue edges} \label{inter}

For any red vertex $\redcol{u}$, if a blue edge $(u',v)$ reachable from $u$ is separated $u$ by a frame, then one of the best red vertices on that frame can reach $u'$.  So let each red vertex remember the best $\leq4$ blue edges it can reach on its own frame.  Then we can define $4$ bitmasks $\{B^\beta(u)\}_{0\leq\beta\leq3}$ such that for any $i$ finding the highest 1-bit $\leq{}i$ in each, gives at most $4$ levels such that the best red vertices reachable from $u$ on those levels together know the best blue edges for $u$.